%% file: main.tex
\documentclass[11pt]{article}
\usepackage{amssymb, amsmath, amsthm}
\usepackage{enumitem}
\usepackage{fullpage}
\usepackage{color}
\usepackage{url}

\makeatletter
\newtheorem*{rep@theorem}{\rep@title}
\newcommand{\newreptheorem}[2]{%
\newenvironment{rep#1}[1]{%
 \def\rep@title{#2 \ref{##1}'}%
 \begin{rep@theorem}}%
 {\end{rep@theorem}}}
\makeatother

\theoremstyle{plain}
\newtheorem{theorem}{Theorem}
\newreptheorem{theorem}{Theorem}
\newtheorem{lemma}[theorem]{Lemma}
\newtheorem{corollary}[theorem]{Corollary}
\newtheorem{proposition}[theorem]{Proposition}

\theoremstyle{definition}

\newcommand{\R}{\mathbb{R}}
\newcommand{\p}{\textrm{P}}
\newcommand{\G}{\mathcal{G}}
\newcommand{\bS}{\mathbb{S}}
\DeclareMathOperator*{\E}{\mathbb{E}}

\DeclareMathOperator{\diag}{diag}

\DeclareMathOperator{\tr}{tr}
\DeclareMathOperator{\poly}{poly}
\DeclareMathOperator{\polylog}{polylog}

\newcommand{\vertiii}[1]{{\left\vert\kern-0.25ex\left\vert\kern-0.25ex\left\vert #1 
    \right\vert\kern-0.25ex\right\vert\kern-0.25ex\right\vert}}
\newcommand{\norm}[1]{{\left\| #1\right\|}}

\title{Embeddings of Schatten Norms with Applications to Data Streams}
\author{
Yi Li\\Division of Mathematics\\
School of Physical \& Mathematical Sciences\\ Nanyang Technological University\\\texttt{yili@ntu.edu.sg}
\and
David P. Woodruff\\IBM Almaden Research Center \\ \texttt{dpwoodru@us.ibm.com}
}
\begin{document}
\maketitle

\begin{abstract}
  Given an $n \times d$ matrix $A$, its Schatten-$p$ norm, $p \geq 1$, is defined as
  $\|A\|_p = \left (\sum_{i=1}^{\textrm{rank}(A)}\sigma_i(A)^p \right )^{1/p}$,
  where $\sigma_i(A)$ is the $i$-th largest singular value of $A$. These norms
  have been studied in functional analysis in the context of non-commutative $\ell_p$-spaces,
  and recently in data stream and linear sketching models of computation.
  Basic questions on the relations between these norms, such as their embeddability, are still
  open. Specifically, given a set of
  matrices $A^1, \ldots, A^{\poly(nd)} \in \mathbb{R}^{n \times d}$, suppose we want to construct a linear map
  $L$ such that $L(A^i) \in \mathbb{R}^{n' \times d'}$ for each $i$, where $n' \leq n$ and $d' \leq d$, and further,  
  $\|A^i\|_p \leq \|L(A^i)\|_q \leq D_{p,q} \|A^i\|_p$ for a given approximation factor $D_{p,q}$ and real number
  $q \geq 1$. Then how large do $n'$ and $d'$ need to be as a function of $D_{p,q}$?

  We nearly resolve this question for every $p, q \geq 1$, for the case where $L(A^i)$ can be expressed
  as $R \cdot A^i \cdot S$, where $R$ and $S$ are arbitrary matrices that are allowed to depend on
  $A^1, \ldots, A^t$, that is,
  $L(A^i)$ can be implemented by left and right matrix multiplication. Namely, for
  every $p, q \geq 1$, we
  provide nearly matching upper and lower bounds on the size of $n'$ and $d'$ as a function of $D_{p,q}$. Importantly,
  our upper bounds are {\it oblivious}, meaning that $R$ and $S$ do not depend on the $A^i$, while our lower bounds hold even if $R$ and $S$ depend on the $A^i$. 
  As an application of our upper bounds,
  we answer a recent open question of Blasiok et al. about space-approximation trade-offs for the Schatten $1$-norm,
  showing in a data stream it is possible to estimate
  the Schatten-$1$ norm up to a factor of $D \geq 1$ using $\tilde{O}(\min(n,d)^2/D^4)$ space. 
\end{abstract}

\input{intro}
\input{prelim}
\input{lowerbound1}
\input{lowerbound2}
\input{upperbound}

\input{streaming}

\bibliographystyle{plainurl}
\bibliography{literature}

\end{document}

%% file: intro.tex
\section{Introduction}\label{sec:intro}
 Given an $n \times d$ matrix $A$, its Schatten-$p$ norm, $p \geq 1$, is defined to be 
  $\|A\|_p = \left (\sum_{i=1}^{\textrm{rank}(A)}\sigma_i(A)^p \right )^{1/p}$,
 where $\sigma_i(A)$ is the $i$-th largest singular value of $A$, i.e., the square root of the $i$-th largest
 eigenvalue of $A^TA$. The Schatten-$1$ norm is
 the nuclear norm or trace norm, the Schatten-$2$ norm is the Frobenius norm, and 
 the Schatten $\infty$-norm, defined as the limit of the Schatten-$p$ norm when $p\to\infty$, is the operator norm. The Schatten $1$-norm has applications
 in non-convex optimization \cite{cr12}, while Schatten-$2$ and
 Schatten-$\infty$ norms are useful in geometry and linear algebra, see, e.g., \cite{w14}. Schatten-$p$ norms
 for large $p$ also provide approximations to the Schatten-$\infty$ norm. 

 The Schatten
 norms appear to be significantly harder to compute or approximate than the vector $\ell_p$-norms in various
 models of computation, and
 understanding the complexity of estimating them has led to new algorithmic ideas and
 lower bound techniques. The main difficulty is that we do not directly have access to the spectrum
 of $A$, and na\"ively it is costly in space and time to extract useful information about it.
 A line of work has focused
 on understanding the complexity of estimating such norms in
 the data stream model with $1$-pass over the stream \cite{lw16a} as well as with
 multiple passes \cite{bcky16}, the sketching model \cite{akr15,lnw14,lw16b}, statistical models \cite{kv16},
 as well as the general RAM model \cite{MNSUW17,UCS16}. Dimensionality reduction in these norms also has applications
 in quantum computing \cite{hms11,w05}. 
 It has also been asked in 
 places if the Schatten-$1$ norm admits non-trivial nearest neighbor search data structures~\cite{a10}. 
 
\paragraph{Our Results.} In this paper we study the embeddability of the Schatten-$p$ norm into the Schatten-$q$ norm for linear maps
 implementable by matrix multiplication. More concretely, we first ask for the following form of embeddability:
 given $n$ and $t$ (where $t = \Omega(\log n)$), what is the smallest value of $D_{p,q}$, which we call the {\it distortion},
 such that there exists a distribution $\mathcal{R}$ on $\R^{t\times n}$ satisfying, for any given $n \times d$ matrix $A$, 
\[
\Pr_{R\sim\mathcal{R}}\left\{\norm{A}_p \leq \norm{RA}_q \leq D_{p,q}\norm{A}_p\right\} \geq 1-\exp(-ct)?
\]
Here $c > 0$ is an absolute constant. We can assume, w.l.o.g., that $n = d$ because we can first apply a so-called {\it subspace embedding} matrix 
(see, e.g., \cite{w14} for a survey)
to the left or to the right of $A$ to preserve each of its
singular values up to a constant factor - we refer the reader to \cite[Appendix C]{lnw14} for this standard
argument. We shall show that $D_{p,q}\gtrsim \hat{D}_{p,q}$, where 
\begin{equation}\label{eqn:distortion_ans}
\hat D_{p,q} = \begin{cases}
				n^{\frac1p-\frac12}/t^{\frac1q-\frac12}, & 1\leq p \leq q \leq 2;\\
				n^{\frac 1p-\frac 12}, & 1\leq p\leq 2 \leq q;\\
				\max\{(n/t)^{\frac12-\frac1p}, t^{\frac1p-\frac1q}\}, & 2\leq p\leq q;\\
				n^{\frac 12-\frac 1p}, & 1\leq q\leq 2 \leq p;\\
				n^{\frac12-\frac1p}/t^{\frac12-\frac1q}, & 2\leq q \leq p;\\
				\max\{(n/t)^{\frac1p-\frac12}, (t/\ln t)^{\frac1q-\frac1p}\}, & 1\leq q \leq p \leq 2,
			\end{cases}
\end{equation}
and the notation $f \gtrsim g$ means $f \geq g/C$ for some constant $C > 0$. The constant $C$ in the $\gtrsim$ notation above depends on $p$ and $q$ only. This distortion is asymptotically tight,
up to logarithmic factors, as we also construct a distribution $\mathcal{R}$ on $t$-by-$n$ matrices for which for any
$n \times d$ matrix $A$, 
\[
\Pr_{R\sim\mathcal{R}}\left\{\norm{A}_p \leq \norm{RA}_q \leq \tilde{D}_{p,q}\left(\log\frac{n}{t}\right)\norm{A}_p\right\} \geq 1-\exp\left(-ct\right),
\]
where $\tilde{D}_{p,q}$ differs from $D_{p,q}$ by a constant or a factor of $\log t$. Specifically,
\begin{equation}\label{eqn:distortion_ans_ub}
\tilde{D}_{p,q}\lesssim \begin{cases} 
			\max\{(n/t)^{\frac1p-\frac12},t^{\frac1q-\frac1p}\}, & 1\leq q\leq p\leq 2;\\
			\hat{D}_{p,q}, &\text{otherwise},
			\end{cases}
\end{equation}
where $\hat{D}_{p,q}$ is given in \eqref{eqn:distortion_ans}. 
Replacing $t$ with $t/(\ln(n/t))$, we arrive at a matching failure probability and distortion, while using a logarithmic factor more number of rows in $R$.
Namely, we construct a distribution $\mathcal{R}$ on matrices with $t\ln(n/t)$ rows for which 
\[
\Pr_{R\sim\mathcal{R}}\left\{\norm{A}_p \leq \norm{RA}_q \leq \tilde{D}_{p,q}\norm{A}_p\right\} \geq 1-\exp\left(-ct\right).
\]
We can also sketch $RA$ on the right by a subspace embedding matrix $S$ with $\Theta(t)$ rows, which yields
\[
\Pr_{R,S}\left\{\norm{A}_p \leq \norm{RAS^T}_q \leq \tilde{D}_{p,q}\norm{A}_p\right\} \geq 1-\exp\left(-ct\right).
\]
We show that this two-sided sketch is asympotically optimal for two-sided sketches in its product of number of rows of $R$ and number of columns of $S$, up to logarithmic factors. Formally, we next ask: what is the smallest value of $D_{p,q}$ for which there exists a distribution $\mathcal{G}_1$ on $\R^{r\times n}$ and a distribution $\mathcal{G}_2$ on $\R^{n\times s}$ satisfying
\[
\Pr_{\substack{R\sim\G_1, S\sim\G_2}}\left\{\norm{A}_p \leq \norm{RAS}_q \leq D_{p,q}\norm{A}_p\right\} \geq 1-\exp(-c\min\{r,s\})?
\]
Again we can assume, w.l.o.g, that $r = s$, because otherwise we can compose $R$ or $S$ with a subspace embedding to 
preserve all singular values up to a constant factor\footnote{That is, if $r\leq s$, we can choose a subspace embedding matrix $H$ of dimension $n\times \Theta(r)$ such that $\norm{RASH}_q = \Theta(\norm{RAS}_q)$ with probability $\geq 1-\exp(-s)$, and then pad $R$ with zero rows so that $R$ has the same number of rows as columns of $S$, increasing the number of rows of $R$ by at most a constant factor.}. Henceforth for the two-sided problem, we assume that $\G_1$ and $\G_2$ are distributions on $\R^{t\times n}$. We also prove a matching lower bound that $D_{p,q}\gtrsim \hat{D}_{p,q}$ except in the case when $1\leq q\leq p\leq 2$, where we instead obtain a matching lower bound up to logarithmic factors, namely, $D_{p,q}\gtrsim \max\{(n/t)^{\frac1p-\frac12}/\log^\frac{3}{2} t, (t/\ln t)^{\frac1q-\frac1p}\}$.

In the important case when $p=q=1$, our results show a space-approximation tradeoff for estimating the Schatten $1$-norm (or trace norm) in a data stream, answering a question posed by Blasiok et al.~\cite{BBCKY16}. This application crucially uses
that $R$ and $S$ are oblivious to $A$, i.e., they can be sampled and succinctly stored without looking at $A$. Specifically, when each entry of $A$ fits in a word of $O(\log n)$ bits, we can choose $R$ and $S$ to be Gaussian random matrices with entries truncated to $O(\log n)$ bits and with entries drawn from
a family of random variables with bounded independence (see Appendix~\ref{sec:streaming}). For time-efficiency purposes, $R$ and $S$ can also be chosen
to be Fast
Johnson Lindenstrauss Transforms or sparse embedding matrices \cite{clarkson2013low,mm13,nn13}, though they will have larger dimension, especially
to satisfy the exponential probability of failure in the problem statement
(and even with constant failure probability, the dimension will be slightly
larger; see \cite{w14} for a survey). 

Choosing $R$ and $S$ to be Gaussian matrices, our result provides 
a data stream algorithm using $(n^2/D^4)\polylog(n)$ bits of memory, and achieving approximation factor $D$ (taking $t = n/D^2$).
While $\|A\|_2$, the Frobenius norm of $A$, provides a $\sqrt{n}$-approximation to $\|A\|_1$ and can be approximated up to a constant factor in a data
stream using $O(1)$ words of space, if we want an algorithm achieving a better approximation factor then all that was known
was an algorithm requiring $O(n^2)$ words of space, namely, the trivial algorithm of storing $A$ exactly and achieving
$D = 1$. It was asked in
\cite{BBCKY16} if there is a smooth trade-off between the case when $D = 1$ and $D = \sqrt{n}$; our
$(n^2/D^4)\polylog(n)$ space algorithm provides the first such trade-off, and is optimal at the two extremes. Our results
are the first of their kind for large approximation factors $D \gg 1$ for estimating the Schatten-$p$ norms in a data
stream. 

Finally, while in our upper bounds $R$ and $S$ are chosen obliviously to $A$, for our lower bounds we would like to rule out
those $R$ and $S$ which are even allowed to depend on $A$. Clearly, if there is only a single matrix $A$, this question is
ill-posed as one can just choose $R$ and $S$ to have a single row and column so that $\|RAS\|_q = \|A\|_p$. Instead, we ask
the question analogous to the Johnson-Lindenstrauss transform (see e.g., \cite{LN16}): given $A^1, \ldots, A^{\poly(n)}$, can we construct an $R$ with
$t$ rows and an $S$ with $t$ columns for which $\|A^i\|_p \leq \|RA^iS\|_q \leq D_{p,q} \|A^i\|_p$ for all $i$? We show that our
lower bound on the trade-off between $D_{p,q}$ and $t$ given by \eqref{eqn:distortion_ans} continues to hold even in this setting.

\paragraph{Our Techniques.} We shall focus on the case $p=q$ in this description of our technical overview. For our upper bounds, a natural idea is to take $R$ to be a (normalized) Gaussian random matrix, and the analysis of the quantity $\|RA\|_p$, when $p\geq 2$, follows fairly directly from the so-called non-commutative Khintchine inequality as follows. 

\begin{lemma}[Non-commutative Khintchine Inequality \cite{LP86}] Suppose that $C_1,\dots,C_n$ are (deterministic) matrices of the same dimension and $g_1,\dots,g_n$ are independent $N(0,1)$ variables. It holds that
\[
\E_{g_1,\dots,g_n} \norm{\sum_i g_i C_i}_p \simeq \max\left\{ \norm{\left(\sum_i C_i C_i^T\right)^{\frac12}}_p, \norm{\left(\sum_i C_i^T C_i\right)^{\frac12}}_p \right\},\quad p\geq 2.
\]
\end{lemma}
In order to estimate $\|RA\|_p$, we can write
\[
RA = \sum_{i,j} r_{ij} (e_ie_j^T A) =: \sum_{i,j} r_{ij} C_{ij}
\]
and it is straightforward to compute that
\[
\sum_{i,j} C_{ij}C_{ij}^T 
= \tr(AA^T)I_t = \|A\|_F^2 I_t,\quad 
\sum_{i,j} C_{ij}^T C_{ij} = 
t\cdot A^T A.
\]
It follows from the non-commutative Khintchine inequality that (recall that $R$ is a normalized Gaussian matrix with $N(0,1/t)$ entries)
\[
\E\norm{RA}_p \simeq \max\left\{ t^{\frac12-\frac1p}\|A\|_F, \|A\|_p \right\},\quad p\geq 2.
\]
Using a concentration inequality for Lipschitz functions
on Gaussian space, one can show that $\norm{RA}_p$ is concentrated around
$\E\norm{RA}_p$, and using standard the standard relationship between
$\|A\|_F$ and $\|A\|_p$ then completes the argument. 

When $p<2$, the non-commutative Khintchine inequality gives a much less tractable characterization, so we need to analyze $\|RA\|_p$ in a different manner, which is potentially of independent interest. Our analysis also works for non-Gaussian matrices $R$ whenever $R$ satisfies certain properties, which, for instance, are satisfied by a Fast Johnson-Lindenstrauss Transform. 

\vskip 0.5ex

\noindent\textbf{\textsf{Upper bound.}} We give an overview of our upper bound now, focusing on the one-sided case, since the two-sided case follows by simply right-multiplying by a generic subspace embedding $S$. Here we focus on the case in which $R$ is an $r\times n$ Gaussian matrix, where $r = t \cdot \polylog(n)$. By rotational invariance of Gaussian matrices, and for the purposes of computing $\|AR\|_p$, we can assume that $A$ is diagonal. Let $A_1$ be the restriction of $A$ to its top $\Theta(t \log n)$ singular values. Since $R$ is a Gaussian matrix with at least $t\log n$ rows, it is well-known that $R$ is also a subspace embedding on $A_1$ (see, e.g., \cite[Corollary 5.35]{V12}), namely, $\sigma_i(RA_1)\simeq \sigma_i(A_1)$ for all $i$, and thus $\|RA_1\|_p \simeq \|A_1\|_p = \Omega(\|A\|_p)$ when $\|A_1\|_p = \Omega(\|A\|_p)$. 

If it does not hold that $\|A_1\|_p = \Omega(\|A\|_p)$, then the singular values of $A$ are ``heavy-tailed'', and we show how to find a $\sigma_i(A)$ with $i < \Theta(t\log n)$ for which $\sigma_i^2(A)$ is relatively small compared to $\sigma_i^2(A)+\sigma_{i+1}^2(A)+\cdots+\sigma_n^2(A)$. More specifically, let $A_2$ be the restriction of $A$ to $\sigma_i(A),\dots,\sigma_n(A)$. Then we have that $\|A_2\|_{op}\lesssim \|A_2\|_F / \sqrt{t}$. Since for a Gaussian matrix $R$ it holds that $\|RA_2\|_{op}\lesssim \|A_2\|_{op} + \|A\|_F/\sqrt{r}$ (see Proposition~\ref{prop:gaussian_sketch_op}), we thus have that $\|RA_2\|_{op}\lesssim \|A_2\|_F/\sqrt{t}$. On the other hand, $\|RA_2\|_{F}\simeq \|A_2\|_F$. This implies there exist $\Omega(t)$ singular values of $RA_2$ that are $\Omega(\|A_2\|_F/\sqrt{t})$, which yields that $\|RA_2\|_p\gtrsim \|A_2\|_p = \Omega(\|A\|_p)$.
Therefore we have established the lower bound that $\|RA\|_p\geq \max\{\|RA_1\|_p,\|RA_2\|_p\}$ in terms of $\|A\|_p$.

To upper bound $\|RA\|_p$ in terms of $\|A\|_p$, note that $\|RA\|_p\leq \|RA_1\|_p + \|RA_2\|_p$ by the triangle inequality, where $A_1,A_2$ are as above. Again it follows from the subspace embedding property of $R$ that $\|RA_1\|_p\lesssim \|A_1\|_p \leq \|A\|_p$. Regarding $\|RA_2\|_p$, we relate its Schatten-$p$ norm to its Frobenius norm and use the fact that $\|RA_2\|_F\simeq \|A_2\|_F$. This gives an upper bound of $\|RA_2\|_p$ in terms of $\|A_2\|_p$, and using that $\|A_2\|_p \leq \|A\|_p$, it gives an upper bound in terms of $\|A\|_p$. This is sufficient to obtain an overall upper bound on $\|RA\|_p$.

\vskip 0.5ex
\noindent\textbf{\textsf{Lower bound.}} Now we give an overview of our lower bounds for some specific cases. First consider one-sided sketches. We choose our hard distribution as follows: we choose an $n\times (10t)$ Gaussian matrix $G$ padded with $0$s to become an $n\times n$ matrix. For a sketch matrix $R$ containing $t$ rows, by rotational invariance of Gaussian matrices, $\|RG\|_p$ is identically distributed to $\|\Sigma_R G'\|_p$, where $\Sigma_R$ is the $t\times t$ diagonal matrix consisting of the singular values of $R$, and where $G'$ is a $t\times (10t)$ Gaussian matrix. It is a classical result that all singular values of $G'$ are $\Theta(\sqrt{t})$ and thus $\|RG\|_p\simeq \sqrt{t}\|R\|_p$. This implies that 
\begin{equation}\label{eqn:ourtech1}
\sqrt{n}t^{\frac 12-\frac 1p}\lesssim \sqrt{t}\|R\|_p \lesssim D_{p,p}\sqrt{n}t^{\frac 12-\frac 1p},
\end{equation}
since all non-zero singular values of $G$ are $\Theta(\sqrt{n})$. On the other hand, applying $R$ to the $n\times n$ identity matrix gives that
\begin{equation}\label{eqn:ourtech2}
n^{\frac 1p}\leq \|R\|_p\leq D_{p,p}n^{\frac 1p}.
\end{equation}
Combining \eqref{eqn:ourtech1} and \eqref{eqn:ourtech2} gives that $D_{p,p}\geq \max\{(n/t)^{1/2-1/p},(n/t)^{1/p-1/2}\}$.

For the two-sided sketch, we change the hard distribution to (i) $n\times n$ Gaussian random matrix $F$ and (ii) the distribution of $GH^T$, where $G$ and $H$ are $n\times \Theta(t)$ Gaussian random matrices. The proof then relies on the analysis for $\norm{RFS^T}_p$ and $\norm{RGH^TS^T}_p$. When $p\geq 2$, non-commutative Khintchine inequality gives immediately that 
\begin{equation}\label{eqn:ourtech6}
\norm{RGH^TS^T}_p\simeq \sqrt{t}\norm{RFS^T}_p\simeq \sqrt{t}\max\{\norm{R}_p\norm{S}_{op},\norm{R}_{op}\norm{S}_p\},\quad p\geq 2.
\end{equation}
When $p<2$, a different approach is followed. We divide the singular values of $R$ and $S$ into bands, where each band contains singular values within a factor of 2 from each other. We shall consider the first $\Theta(\log t)$ bands only because the remaining singular values are $1/\poly(t)$ and negligible. Now, if all singular values of $R'$ and $S'$ are within a factor of $2$ from each other, then $\norm{R'F(S')^T}_p\simeq \norm{R'}_{op}\norm{S'}_{op}\norm{F}_p$ and $\norm{R'GH^T(S')^T}_p\simeq \norm{R'}_{op}\norm{S'}_{op}\norm{GH^T}_p$. It is not difficult to see that
\begin{equation}\label{eqn:ourtech3}
\|GH^T\|_p\simeq\sqrt{t}\|F\|_p
\end{equation}
Since $R'$ and $S'$ consist of one of the $\Theta(\log t)$ bands of $R$ and $S$, respectively, it follows that 
\begin{equation}\label{eqn:ourtech4}
\norm{RGH^TS^T}_p\simeq \sqrt{t}/\polylog(t) \cdot \norm{RFS^T}_p,\quad p < 2.
\end{equation}
A lower bound of $D_{p,p}$ then follows from combining \eqref{eqn:ourtech3}, \eqref{eqn:ourtech6} (or \eqref{eqn:ourtech4}) with
\[
\norm{F}_p \leq \norm{RFS^T}_p \leq D_{p,p}\norm{F}_p,\quad\text{and}\quad 
\norm{GH^T}_p \leq \norm{RGH^TS^T}_p \leq D_{p,p}\norm{GH^T}_p.
\]
To strengthen the lower bound for the sketches that even depend on the input matrix, we follow the approach in \cite{LN16}. We first work with random hard instances, and then sample input matrices $A^1,\dots,A^{\poly(n)}$ from the hard distribution, and apply a net argument on sketching matrices $R$ and $S$ to obtain a deterministic statement, which states that for any fixed $R$ and $S$ such that the distortion guarantee is satisfied with all samples $A^1,\dots,A^{\poly(n)}$, the distortion lower bound remains to hold.

%% file: prelim.tex
\section{Preliminaries}

\paragraph{Notations.} Throughout the paper, we use $f\lesssim g$ to denote $f\leq Cg$ for some constant $C$, $f\gtrsim g$ to denote $f\geq Cg$ for some constant $C$ and $f\simeq g$ to denote $C_1 g\leq f \leq C_2 g$ for some constants $C_1$ and $C_2$.

\paragraph{Bands of Singular Values.} Given a matrix $A$, we split the singular values of $A$, $\sigma_1(A)\geq \sigma_2(A)\geq \cdots$, into bands such that the singular values in each band are within a factor of 2 from each other. Formally, define the $i$-th singular value band of $A$ as 
\[
\mathcal{B}_i(A) = \left\{k: \frac{\|A\|_{op}}{2^{i+1}} < \sigma_k(A) \leq \frac{\|A\|_{op}}{2^{i}}\right\}, \quad i\geq 0,
\]
and let $N_i(A) = \left|\mathcal{B}_i(A)\right|$, the cardinality of the $i$-th band.

\paragraph{Extreme Singular Values of Gaussian Matrices.} We shall repeatedly use the following results on Gaussian matrices.
\begin{proposition}[{\cite[Corollary 5.35]{V12}}]\label{prop:gaussian_matrix} Let $G$ be an $r\times n$ ($r<n$) Gaussian random matrix of i.i.d.\ entries $N(0,1)$. With probability at least $1-2\exp(-u^2/2)$, it holds that
\[
\sqrt{n}-\sqrt{r}-u \leq s_{\min}(G) \leq s_{\max}(G) \leq \sqrt{n}+\sqrt{r}+u.
\]
\end{proposition}

Combining \cite[Corollary 3.21]{LT91} and the concentration bound in Gauss space \cite[Proposition 5.34]{V12}, we also have
\begin{proposition}\label{prop:gaussian_sketch_op} Let $A$ be a deterministic $n\times n$ matrix and $G$ be an $r\times n$ ($r<n$) Gaussian random matrix of i.i.d.\ entries $N(0,1)$. Then 
\[
\|GA\|_{op} \leq K(\|A\|_{op}\sqrt r + \|A\|_F)
\]
with probability at least $1-\exp(-c\sqrt{K}r)$, where $c > 0$ is an absolute constant.
\end{proposition}

\paragraph{Nets on Matrices.} The following fact is used in \cite{LN16}, which shows the lower bound for the target dimension of linear space embedding.
\begin{proposition}[{\cite[Lemma 2]{LN16}}] There exists a net $\mathcal{R}\subset \bigcup_{t=1}^{t_0} \R^{t\times n}$ of size $\exp(O(t_0n\ln(Dn/\eta))$ such that for any $R\in \R^{t\times n}$ ($1\leq t\leq t_0$) with column norms in $[1,D]$, we can find $R'\in \mathcal{R}$ such that $\|R-R'\|_{op}\leq \eta$.
\end{proposition}

%% file: lowerbound1.tex
\section{Lower bounds For One-sided Sketches}\label{sec:lb_one_sided}
\begin{theorem}
Let $1\leq p\leq 2$ and $q\geq 1$. There exist a set $T\subset \R^{n\times n}$ with $|T|=O(n)$ and an absolute constant $c\in (0,1)$ such that, if it holds for some matrix $R\in \R^{t\times n}$ with $t\leq cn$ and for all $A\in T$ that
\begin{equation}\label{eqn:distortion}
\norm{A}_p \leq \norm{RA}_q \leq D_{p,q}\norm{A}_p 
\end{equation}
it must hold that
\[
D_{p,q} \gtrsim \begin{cases}
				n^{\frac1p-\frac12}/t^{\frac1q-\frac12}, & q \leq 2\\
				n^{\frac 1p-\frac 12}, & q \geq 2
			\end{cases}
\]
\end{theorem}
\begin{proof}
First we consider the case $1\leq p\leq q\leq 2$. We take $T = \{I_n,e_1e_1^T,\dots,e_ne_n^T\}$, where $I_n$ is the identity matrix and $\{e_i\}$ the canonical basis in $\R^n$. 

We know from letting $A=I$ in \eqref{eqn:distortion} that
\[
\|R\|_q \geq n^{1/p}.
\]
On the other hand,
\[
\|R\|_q \leq t^{\frac1q-\frac12}\|R\|_F.
\]
hence
\[
\|R\|_F^2 \geq \frac{n^{\frac2p}}{t^{\frac2q-1}}.
\]
Hence there exists $i$ such that the $i$-th column of $R$, denoted by $R_i$, satisfies that
\[
\|R_i\|_2^2 \geq \frac{n^{\frac2p-1}}{t^{\frac2q-1}}.
\]
Letting $A=e_ie_i^T$ in \eqref{eqn:distortion},
\[
D_{p,q} = D_{p,q}\|e_ie_i^T\|_p\geq \|Re_ie_i^T\|_q = \|R_i\|_2 \geq \frac{n^{\frac1p-\frac{1}{2}}}{t^{\frac1q-\frac{1}{2}}}.
\]

A similar argument works for $1\leq p<2<q$. We take the same $T$ as above. And now $n^{1/p}\leq \|R\|_q\leq \|R\|_F$ so there exists $i$ such that $\|R_i\|_2\geq n^{2/p-1}$. Letting $A=e_ie_i^T$ yields $D_{p,q}\geq n^{\frac{1}{p}-\frac12}$.
\end{proof}

\begin{theorem}\label{thm:p>2>q}
Let $p\geq 2$ and $q\geq 1$. There exist a set $T\subset \R^{n\times n}$ with $|T|=O(n)$ and an absolute constant $c\in (0,1)$ such that, if it holds for some matrix $R\in \R^{t\times n}$ with $t\leq cn$ and for all $A\in T$ that
\begin{equation}
\norm{A}_p \leq \norm{RA}_q \leq D_{p,q}\norm{A}_p \tag{\ref{eqn:distortion}}
\end{equation}
it must hold that
\[
D_{p,q} \gtrsim \begin{cases}
				n^{\frac12-\frac1p}, &q\leq 2;\\
				n^{\frac12-\frac1p}/t^{\frac12-\frac1q}, &q\geq 2.
			  \end{cases}	
\]
\end{theorem}
\begin{proof}
Let $T=\{I_n,e_1e_1^T,\dots,e_ne_n^T\}$. In \eqref{eqn:distortion}, take $A=I_n$,
\[
\norm{R}_q\leq D_{p,q} n^{\frac 1p}.
\]
Take $A=e_ie_i^T$,
\[
\norm{R_i}_2 = \norm{Ge_ie_i^T}_q \geq 1,
\]
where $R_i$ is the $i$-th column of $R$, and hence $\norm{R}_F\geq \sqrt{n}$.
The lower bound for $D_{p,q}$ follows from the facts
\[
\norm{R}_F\leq \begin{cases}
				\norm{R}_q, & q < 2;\\
				t^{\frac12-\frac1q}\norm{R}_q, & q > 2.
		\end{cases}\qedhere
\]
\end{proof}

\begin{theorem} \label{thm:lb_on_t}
Let $p,q\geq 1$. There exist a set $T\subset \R^{n\times n}$ with $|T|=\exp(O(t))$ and an absolute constant $c\in (0,1)$ such that, if it holds for some matrix $R\in \R^{t\times n}$ with $t\leq cn$ and for all $A\in T$ that
\begin{equation}
\norm{A}_p \leq \norm{RA}_q \leq D_{p,q}\norm{A}_p \tag{\ref{eqn:distortion}}
\end{equation}
it must hold that
\[
D_{p,q} \gtrsim t^{\frac1p-\frac1q}.
\]
\end{theorem}
\begin{proof}
Let $\mathcal{N}$ be an $\epsilon$-net on the unit sphere $\bS^{t-1}$ of size $(1+2/\epsilon)^t$. For each $x\in \R^t$, define \[
A_x = \begin{pmatrix}
		x & 0 \\
		0 & 0
	\end{pmatrix}	
\]
then $\|RA_x\|_p = \|R'x\|_2$ and $\|A_x\|_p = \|x\|_2$ for all $p$, where $R'$ is the left $t\times t$ block of $R$. Letting $A=A_x$ in \eqref{eqn:distortion},
\[
\|x\|_2\leq \|R'x\|_2\leq D_{p,q}\|x\|_2,\quad \forall x\in \mathcal{N}.
\]
Let
\[
A_t=\begin{pmatrix} I_t & 0\\ 0 & 0\end{pmatrix},
\]
and take $T=\{A_t\}\cup\mathcal{N}$. Since $R$ satisfies \eqref{eqn:distortion} on $\mathcal{N}$, a standard argument (see, e.g., \cite[p233]{V12}) shows that $\|R'\|_{op} \leq D_{p,q}/(1-\epsilon)$. Thus $\|R'\|_q \lesssim D_{p,q} t^{1/q}$. On the other hand, letting $A=A_t$ in \eqref{eqn:distortion} gives that $\|R'\|_q \geq t^{1/p}$. Hence $D_{p,q}\gtrsim t^{1/p-1/q}$.
\end{proof}

\begin{theorem} \label{thm:lb2_on_t}
Let $p,q\geq 1$. There exist a set $T\subset \R^{n\times n}$ with $|T|=\exp(O(t))$ and an absolute constant $c\in (0,1)$ such that, if it holds for some matrix $R\in \R^{t\times n}$ with $t\leq cn$ and for all $A\in T$ that
\begin{equation}
\norm{A}_p \leq \norm{RA}_q \leq D_{p,q}\norm{A}_p \tag{\ref{eqn:distortion}}
\end{equation}
it must hold that
\[
D_{p,q} \gtrsim (t/\ln t)^{\frac1q-\frac1p}.
\]
\end{theorem}
\begin{proof}
Let $t' = t/\ln t$. Assume that $D = D_{p,q}\leq t'$, otherwise the result holds already. We can further assume that $p>q$. Let $\mathcal{N}$ be a $(1/2D)$-net on $\bS^{t'-1}$ of size $\exp(\Theta(t'\ln D)) = \exp(O(t'\ln t')) = \exp(O(t))$. Proceed as in the proof of Theorem~\ref{thm:lb_on_t}, and we arrive at
\[
\|x\|_2\leq \|R'x\|_2\leq D_{p,q}\|x\|_2,\quad \forall x\in \mathcal{N},
\]
where $R'$ is the left $t'$ columns of $R$. The proof of Theorem~\ref{thm:lb_on_t} shows that $\norm{R'}_{op}\leq 2D$. We claim that $s_{\min}(R')\geq 1/2$, or equivalently, $\|R'x\|_2\geq 1/2$ for all $x\in \bS^{t'-1}$. For $x\in \bS^{t'-1}$, find $y\in \mathcal{N}$ such that $\|x-y\|_2\leq 1/(2D)$, and thus
\[
\|R'x\|_2 \geq \|R'y\|_2 - \|R'(x-y)\|_2 \geq 1 - \norm{R'}_{op}\|x-y\|_2\geq 1 - D\cdot\frac{1}{2D} = \frac12.
\]
We then have that $\|RA_t\|_q = \|R'\|_q \lesssim D(t')^{1/p}$ and $\|RA_{t'}\|_q = \|R'\|_q\gtrsim (t')^{1/q}$, and it follows that $D\gtrsim (t')^{1/q-1/p}$.
\end{proof}

\begin{theorem} \label{thm:p>2}
Let $p>2$ and $p > q$. There exist a set $T\subset \R^{n\times n}$ with $|T|=\poly(n)$ and an absolute constant $c\in (0,1)$ such that, if it holds for some matrix $R\in \R^{t\times n}$ with $t\leq cn$ and for all $A\in T$ that
\begin{equation}
\norm{A}_p \leq \norm{RA}_q \leq D_{p,q}\norm{A}_p \tag{\ref{eqn:distortion}}
\end{equation}
it must hold that
\begin{equation}\label{eqn:distortion_ans_2}
D_{p,q} \gtrsim (n/t)^{\frac12-\frac1p}
\end{equation}
\end{theorem}

Instead of proving this theorem, we prove the following rephrased version.

\begin{reptheorem}{thm:p>2}[rephrased]\label{thm:p>2_rephrased}
Let $p>2$ and $p>q$. There exist an absolute constant $D_0$ and a set $T\subset \R^{n\times n}$ with $|T|=O(n\ln(Dn))$ such that, if $D\geq D_0$ and it holds for some matrix $R\in \R^{t\times n}$ and for all $A\in T$ that
\begin{equation}\label{eqn:distortion'}
\norm{A}_p \leq \norm{RA}_q \leq D^{\frac12-\frac1p}\norm{A}_p 
\end{equation}
then it must hold that $t\gtrsim n/D$.
\end{reptheorem}

\begin{proof}
Let $r=n/(\rho^2D)$ and $t_0 = \theta r$ for some constants $\rho>1$ and $\theta\in (0,1)$ to be determined. We shall show that if $t\leq t_0$, it will not happen that $R$ satisfies \eqref{eqn:distortion'} for all $A\in T$.

Let $\mathcal{D}$ be the distribution of Gaussian random matrices of dimension $n\times r$ with i.i.d.\ entries $N(0,1/r)$. Let $R=U\Sigma V^T$ be the singular value decomposition of $R$ and $A\sim \mathcal{D}$. Then by rotational invariance of the Schatten norm and Gaussian random matrices, we know that $\|RA\|_q$ is identically distributed as $\|\Sigma A\|_q = \|B^T\Sigma'\|_q$, where $\Sigma'$ is the left $t\times t$ block of $\Sigma$ and $B$ is formed by the first $t$ rows of $A$. 

 It follows from Proposition~\ref{prop:gaussian_matrix} that with probability $\geq 1-\exp(-c_1c_2r)$,
\[
s_{\max}(B)\leq 1+2c_1\sqrt{\frac{t}{r}} \leq 1+2\sqrt{\theta} c_1,
\]
and thus
\[
\|B^T\Sigma'\|_q \leq s_{\max}(B)\|\Sigma'\|_q \leq (1+2\sqrt{\theta}c_1) \|\Sigma'\|_q = 
(1 + \sqrt{\theta}c_1) \|R\|_q \leq (1 + 2\sqrt{\theta}c_1) D^{\frac 12-\frac 1p}n^{\frac1p},
\]
that is, with probability $\geq 1-\exp(-c_1c_2r)$,
\[
\|RA\|_q \leq (1+2\sqrt{\theta}c_1) D^{\frac 12-\frac 1p}n^{\frac1p}.
\]
On the other hand, with probability $\geq 1-\exp(-c_1c_2r)$, all singular values of $A$ are at least $\sqrt{n/r}-2c_1 = \rho\sqrt{D}-2c_1 \geq (1-\epsilon)\rho\sqrt{D}$ if we choose $D_0\geq 4c_1^2/\epsilon^2$. Then
\[
\|RA\|_q \geq \|A\|_p \geq (1-\epsilon)sr^{\frac1p}\sqrt{D} = (1-\epsilon)\rho^{1-\frac2p}n^{\frac1p}D^{\frac12-\frac1p}.
\]
Also, with probability $\geq 1-\exp(-c_1c_2r)$, all singular values of $A$ are at most $\sqrt{n/r}+2c_1 = \rho\sqrt{D}+2c_1 \leq (1+\epsilon)\rho\sqrt{D}$ and thus
\[
\|A\|_p \leq r^{\frac1p}(1+\epsilon)s\sqrt{D} = (1+\epsilon)\rho^{1-\frac{2}{p}}n^{\frac1p}D^{\frac12-\frac1p}.
\]
This motivates the following definitions of constraints for $R\in \R^{t\times n}$ and $A\in\R^{n\times n}$:
\begin{align*}
	\p_1(R,A):\ &\|RA\|_q \leq (1+2\sqrt{\theta}c_1) D^{1/2-1/p} n^{1/p}\\
	\p_2(R,A):\ &\|RA\|_q \geq (1-\epsilon)\rho^{1-\frac2p}n^{1/p}D^{1/2-1/p}\\
	\p_3(A):\ &\|A\|_p \leq (1+\epsilon)\rho^{1-\frac2p}n^{1/p}D^{1/2-1/p}.
\end{align*}
Now, for $m$ samples $A_1,\dots,A_m$ drawn from $\mathcal{D}$, it holds for any fixed $R$ that
\begin{equation}\label{eqn:pre-net}
\Pr_{A_1,\dots,A_m}\left\{\exists i\text{ s.t. }
\p_1(R,A)\text{ and }\p_2(R,A)\text{ and }\p_3(A)\text{ hold}
\right\} \geq 1-e^{-c_1c_2mr}.
\end{equation}

Since $1\leq \|Ge_ie_i^T\|_q \leq D$ and $\|Ge_ie_i^T\|_q = \|R_i\|_2$, we can restrict the matrix $R$ to matrices with column norm in $[1,D]$. Thus we can find a net $\mathcal{R}\subset \bigcup_{t=1}^{t_0} \R^{t\times n}$ of size $\exp(O(t_0n\ln(Dn/\eta))$ such that for any $R$ with column norms in $[1,D]$, we can find $R'\in \mathcal{R}$ such that $\|R-R'\|_{op}\leq \eta$.

Now it follows from \eqref{eqn:pre-net} that
\begin{multline*}
\Pr_{A_1,\dots,A_m}\left\{\forall R\in\mathcal{R}, \exists i,\ \p_1(R,A)\text{ and }\p_2(R,A)\text{ and }\p_3(R,A)\text{ hold}
\right\}\\
 \geq 1-\exp\left(O\left(t_0n\ln\frac{Dn}{\eta}\right)\right)\exp\left(-\frac{c_1c_2}{D}mn\right) > 0,
\end{multline*}
if we choose $m = \Theta(n\ln(Dn))$. Fix $A_1,\dots,A_m$ such that for each $R\in\mathcal{R}$ there exists $i$ such that $\p_1(R',A_i)$ and $\p_2(R',A_i)$ and $\p_3(A_i)$ all hold.

Take $T=\{I_n,e_1e_1^T,\dots,e_ne_n^T,A_1,\dots,A_m\}$. We know that if $R$ satisfies \eqref{eqn:distortion} for all $A\in T$, then there exists $R'$ such that $\|R'-R\|_F\leq \eta$, and there exists $1\leq i\leq m$ such that $\p_1(R',A_i)$, $\p_2(R',A_i)$ and $\p_3(A_i)$ all hold. It follows that
\begin{align*}
\|RA_i\|_q \leq \|R'A_i\|_q + \|(R-R')A_i\|_q &\leq \|R'A_i\|_q + \|R-R'\|_{op}\|A_i\|_p\\
&\leq \left(1+2\sqrt{\theta}c_1+(1+\epsilon)\rho^{1-\frac2p}\eta\right)D^{\frac12-\frac1p}n^{\frac1p}
\end{align*}
and
\begin{align*}
\|RA_i\|_q \geq \|RA_i\|_q - \|(R-R')A_i\|_q &\geq \|R'A_i\|_q - \|R-R'\|_{op}\|A_i\|_p\\
&\geq \left((1-\epsilon)-(1+\epsilon)\eta\right)\rho^{1-\frac2p}D^{\frac12-\frac1p}n^{\frac1p}
\end{align*}
We meet a contradiction when $\theta$, $\epsilon$ and $\eta$ are all sufficiently small and $\rho$ is sufficiently large, for instance, when $\eta=\Theta(\epsilon)$, $\theta=\Theta(\epsilon^2/c_2^2)$ and $\rho=\Theta(1+p\epsilon/(p-2))$.
\end{proof}

Using almost the exact argument with the identical set $T$ as in the proof of Theorem~\ref{thm:p>2} we can prove a similar bound for $p<2$. The proof is omitted.
\begin{theorem}[$p<2$] \label{thm:p<2}
Let $1\leq p<2$ and $p<q$. There exist a set $T\subset \R^{n\times n}$ with $|T|=\poly(n)$ and an absolute constant $c\in (0,1)$ such that, if it holds for some matrix $R\in \R^{t\times n}$ with $t\leq cn$ and for all $A\in T$ that
\begin{equation}
\norm{A}_p \leq \norm{RA}_q \leq D_{p,q}\norm{A}_p \tag{\ref{eqn:distortion}}
\end{equation}
it must hold that
\[
D_{p,q} \gtrsim (n/t)^{\frac1p-\frac12}.
\]
\end{theorem}

%% file: lowerbound2.tex
\section{Lower Bounds for Two-sided Sketches}\label{sec:lb_two_sided}

\begin{theorem}\label{thm:two_sided_1<p<2}
Let $1\leq p\leq 2$ and $q\geq 1$. There exist a set $T\subset \R^{n\times n}$ with $|T|=O(n^2\ln n)$ and an absolute constant $c\in (0,1)$ such that, if it holds for some matrices $R,S\in \R^{t\times n}$ with $t\leq cn$ and for all $A\in T$ that
\begin{equation}\label{eqn:two_sided_distortion}
\norm{A}_p \leq \norm{RAS^T}_q \leq D_{p,q}\norm{A}_p
\end{equation}
it must hold that
\[
D_{p,q} \gtrsim \begin{cases}
				n^{\frac1p-\frac12}/t^{\frac1q-\frac12}, & q \leq 2\\
				n^{\frac 1p-\frac 12}, & q \geq 2
			\end{cases}
\]
\end{theorem}
\begin{proof}
First we consider the case $1\leq p\leq q\leq 2$. Without loss of generality we can assume that maximum column norm of $R$ and $S$ are the same. Let $\mathcal{F}$ be the distribution of $n\times n$ Gaussian matrices of i.i.d.\ entries $N(0,1)$ and $F\sim \mathcal{F}$. With probability $\geq 0.9$, the following conditions hold:
\begin{align*}
	\p_1(F):\ &\|F\|_p \gtrsim n^{1/p}\sqrt{n},\\
	\p_2(F,R,S):\ &\norm{RFS^T}_F \lesssim \|R\|_F\|S\|_F
\end{align*}
Let $g\sim N(0,I_n)$, then with probability $\geq 0.9$ the following conditions further holds:
\begin{align*}
	\p_3(g):\ &\|g\|_p \lesssim \sqrt{n},\\
	\p_4(g,R):\ &\norm{Rg}_2 \gtrsim \|R\|_F.
\end{align*}
Therefore for any fixed $R$ and $S$,
\[
\Pr_{\substack{F\sim \mathcal{F}\\g,h\sim N(0,I_n)}} \left\{
\p_1(F), \p_2(F,R,S), \p_3(g), \p_4(g,R), \p_3(h), \p_4(h,S)\text{ all hold}
\right\} \geq 0.7.
\]
Hence if we draw $m$ samples $F_1,\dots,F_m$ from $\mathcal{F}$ and $2m$ samples $g_1,\dots,g_m,h_1,\dots,h_m$ from $N(0,I_n)$, it holds that
\[
\Pr_{F_i,g_i,h_i} \left\{\exists i\text{ s.t. }
\p_1(F_i), \p_2(F_i,R,S), \p_3(g_i), \p_4(g_i,R), \p_3(h_i), \p_4(h_i,S)\text{ all hold}
\right\} \geq 1-(0.3)^m.
\]
Next, we find a net $\mathcal{M}\subset \bigcup_{t=1}^{cn} \R^{t\times n}$ of size $\exp(O(n^2\ln(n/\eta))$ such that for any $M$ with column norms in $[1,\sqrt{n}]$, we can find $M'\in \mathcal{G}$ such that $\|M-M'\|_{op}\leq \eta$. Let $m = \Theta(n^2\ln(n/\eta))$, we can find $F_1,\dots,F_m$ and $g_1,\dots,g_m,h_1,\dots,h_m$ such that for any $R,S\in \mathcal{M}$, there exists $i$ such that
\[
\p_1(F_i), \p_2(F_i,R,S), \p_3(g_i), \p_4(g_i,R), \p_3(h_i), \p_4(h_i,S)\text{ all hold}.
\]
Let $T=\{e_ie_j^T\}_{i,j=1}^n\cup \{F_1,\dots,F_m\} \cup \{g_ih_i^T\}_{i=1}^m$. Now, given any $R, S$ with maximum column norm $\sqrt{n}$, a standard net argument show that those properties above still hold (probably with slightly smaller or larger heading constants) for some $i$. Thus
\begin{equation}\label{eqn:lb1_aux}
\|R\|_F\|S\|_F\gtrsim \|RF_iS^T\|_F \geq \frac{\|RF_iS^T\|_q}{t^{\frac1q-\frac12}} \gtrsim \frac{\norm{F_i}_q}{t^{\frac1q-\frac12}} \gtrsim \frac{n^{\frac 1p+\frac12}}{t^{\frac1q-\frac12}}
\end{equation}
and
\[
\|R\|_F\|S\|_F\lesssim \|Rg_i\|_2\|Sh_i\|_2 = \norm{Rg_i h_i^TS^T}_q \leq D\norm{g_i h_i^T}_2=D\norm{g_i}_2\norm{h_i}_2\leq Dn.
\]
It follows immediately that
\[
D\gtrsim \frac{n^{\frac 1p-\frac12}}{t^{\frac 1q-\frac12}}.
\]
When the maximum norm of $R$ and $S$ is at least $\sqrt{n}$, say, $\|R_i\|_2=\|S_j\|_2 \geq \sqrt{n}$, then $D\geq \|Re_ie_j^TS\|_2 = \|R_i\|_2\|S_j\|_2\geq n$. This completes the proof for $q\leq 2$.

When $q>2$, instead of \eqref{eqn:lb1_aux} we have
\[
\|R\|_F\|S\|_F\gtrsim \|RF_iS^T\|_F \geq \|RF_iS^T\|_q \gtrsim \|F\|_q \gtrsim n^{\frac 1p+\frac12}.
\]
and thus $D\gtrsim n^{\frac1p-\frac12}$.
\end{proof}

\begin{theorem}\label{thm:two_sided_p>2>q}
Let $p\geq 2$ and $q\geq 1$. There exist a set $T\subset \R^{n\times n}$ with $|T|=O(n)$ and an absolute constant $c\in (0,1)$ such that, if it holds for some matrices $R,S\in \R^{t\times n}$ with $t\leq cn$ and for all $A\in T$ that
\begin{equation}
\norm{A}_p \leq \norm{RAS^T}_q \leq D_{p,q}\norm{A}_p \tag{\ref{eqn:two_sided_distortion}}
\end{equation}
it must hold that
\[
D_{p,q} \gtrsim \begin{cases}
				n^{\frac12-\frac1p}, &q\leq 2;\\
				n^{\frac12-\frac1p}/t^{\frac12-\frac1q}, &q\geq 2.
			  \end{cases}	
\]
\end{theorem}
\begin{proof}
The proof is similar to that of Theorem~\ref{thm:two_sided_1<p<2}, except that we need upper bounds for $\p_1(F)$ and $\p_4(g,R)$ and lower bounds for $\p_2(F,R,S)$ and $\p_3(g)$. Details are omitted.
\end{proof}

\begin{theorem} \label{thm:two_sided_lb_on_t}
Let $p,q\geq 1$. There exist a set $T\subset \R^{n\times n}$ with $|T|=\exp(O(t))$ and an absolute constant $c\in (0,1)$ such that, if it holds for some matrices $R, S\in \R^{t\times n}$ with $t\leq cn$ and for all $A\in T$ that
\begin{equation}
\norm{A}_p \leq \norm{RAS^T}_q \leq D_{p,q}\norm{A}_p \tag{\ref{eqn:two_sided_distortion}}
\end{equation}
it must hold that
\[
D_{p,q} \gtrsim t^{\frac1p-\frac1q}.
\]
\end{theorem}
\begin{proof}
Let $\mathcal{N}$ be an $\epsilon$-net on the unit sphere $\bS^{t-1}$ of size $(1+2/\epsilon)^t$. Let
\[
A_t=\begin{pmatrix} I_t & 0\\ 0 & 0\end{pmatrix},
\]
and take $T=\{A_t\}\cup\mathcal{N}$. For each pair $(x,y)\in \R^t$, define \[
A_{x,y} = \begin{pmatrix}
		xy^T & 0 \\
		0 & 0
	\end{pmatrix}	
\]
then $\|RA_{x,y}S^T\|_p = \|R'x\|_2\|S'y\|_2$ and $\|A_{x,y}\|_p = \|x\|_2\|y\|_2$ for all $p$, where $R'$ and $S'$ are the leftmost $t\times t$ block of $R$ and $S$ respectively. Letting $A=A_{x,y}$ in \eqref{eqn:two_sided_distortion},
\[
\|x\|_2\|y\|_2\leq \|R'x\|_2\|S'y\|_2\leq D_{p,q}\|x\|_2\|y\|_2,\quad \forall x,y\in \mathcal{N}.
\]
A standard argument as in \cite[p233]{V12} shows that $\|R'\|_{op}\|S'\|_{op}\leq D_{p,q}/(1-\epsilon)$. Thus $\|R'S'\|_q \lesssim D_{p,q} t^{1/q}$. On the other hand, letting $A=A_t$ in \eqref{eqn:two_sided_distortion} gives that $\|R'S'\|_q \geq t^{1/p}$. Hence $D_{p,q}\gtrsim t^{1/p-1/q}$.
\end{proof}

\begin{theorem} \label{thm:two_sided_lb2_on_t}
Let $p,q\geq 1$. There exist a set $T\subset \R^{n\times n}$ with $|T|=\exp(O(t\ln t))$ and an absolute constant $c\in (0,1)$ such that, if it holds for some matrices $R, S\in \R^{t\times n}$ with $t\leq cn$ and for all $A\in T$ that
\begin{equation}
\norm{A}_p \leq \norm{RAS^T}_q \leq D_{p,q}\norm{A}_p \tag{\ref{eqn:two_sided_distortion}}
\end{equation}
it must hold that
\[
D_{p,q} \gtrsim (t/\ln t)^{\frac1q-\frac1p}.
\]
\end{theorem}
\begin{proof}
Let $t'=t/\ln t$ and assume that $D = D_{p,q}\leq t$, otherwise the result holds already. We can further assume that $p>q$. Let $\mathcal{N}$ be a $(1/2D)$-net on $\R^t$ of size $\exp(O(t\ln D)) = \exp(O(t'\ln t')) = \exp(\Theta(t))$. Proceed as in the proof of Theorem~\ref{thm:two_sided_lb_on_t}, and we arrive at
\[
\|x\|_2\|y\|_2\leq \|R'x\|_2\|S'x\|_2\leq D_{p,q}\|x\|_2\|y\|_2,\quad \forall x,y\in \mathcal{N},
\]
where $R'$ and $H'$ are the left $t'$ columns of $R$ and $S$, respectively.. The proof of Theorem~\ref{thm:two_sided_lb_on_t} shows that $\|R\|_{op}\|S\|_{op}\leq 2D$, whence an argument similar to that in Theorem~\ref{thm:lb2_on_t} shows that $s_{\min}(R')s_{\min}(S')\geq 1/2$. We then have that $\|RA_{t'}S^T\|_q = \|R'{S'}^T\|_q \lesssim D(t')^{1/p}$ and $\|RA_{t'}S^T\|_q = \|R'{S'}^T\|_q\gtrsim (t')^{1/q}$, and it follows that $D\gtrsim (t')^{1/q-1/p}$.
\end{proof}

\begin{theorem} \label{thm:two_sided_p<2}
Let $p<2$. There exist a set $T\subset \R^{n\times n}$ with $|T|=\poly(n)$ and an absolute constant $c\in (0,1)$ such that,  if it holds for some matrices $R,S\in \R^{t\times n}$ with $t\leq cn$ and for all $A\in T$ that
\begin{equation}
\norm{A}_p \leq \norm{RAS^T}_q \leq D_{p,q}\norm{A}_p \tag{\ref{eqn:two_sided_distortion}}
\end{equation}
it must hold that
\begin{equation}\label{eqn:two_sided_distortion_ans_2}
D_{p,q} \gtrsim (n/t)^{\frac1p-\frac12}/\log^{\frac{3}{2}}t.
\end{equation}
\end{theorem}

Instead of proving this theorem, we prove the following rephrased version.

\begin{reptheorem}{thm:two_sided_p<2}[rephrased]
Let $p<2$, $p>q$ and $D\geq D_0$ for some an absolute constant $D_0$. There exists a set $T\subset \R^{n\times n}$ with $|T|=O(n\ln(Dn))$ such that it holds for some matrices $R, S\in \R^{t\times n}$ and for all $A\in T$ that
\begin{equation}\label{eqn:two_sided_distortion'}
\norm{A}_p \leq \norm{RAS^T}_q \leq D^{\frac1p-\frac12}\norm{A}_p 
\end{equation}
then it must hold that $t\gtrsim n/(D\log^{3p/(2-p)}t)$.
\end{reptheorem}

We need two auxiliary lemmata.

\begin{lemma}\label{lem:AXB}
Let $A$ and $B$ be deterministic $n\times n$ matrices and $G$ be a Gaussian random matrix of i.i.d.\ $N(0,1)$ entries. It holds with probability $1-O(1)$ that
\[
\|AGB\|_p\lesssim (\log^{\frac 52} n)(\log \log n) \|A\|_{op} \|B\|_{op} E_p(A,B),
\]
where 
\begin{equation}\label{eqn:E_p}
E_p(A,B) = \max_{0\leq i,j\leq 3\log n}\frac{1}{2^{i+j}}\cdot \min\left\{N_i(A),N_j(B)\right\}^{\frac 1p} \cdot \max\left\{\sqrt{N_i(A)}, \sqrt{N_j(B)}\right\}.
\end{equation}
\end{lemma}
\begin{proof}
By rotational invariance we may assume that $A$ and $B$ are diagonal. Write $A = \diag(a_1,\dots,a_n)$ and $B=\diag(b_1,\dots,b_n)$, where $a_i, b_i\geq 0$. By scaling we further assume that $\|A\|_{op}=1$ and $\|B\|_{op}=1$. For notational simplicity, let $I_i = \mathcal{B}_i(A)$, $J_j = \mathcal{B}_j(A)$ and $s_i = |I_i|$, $t_j = |J_i|$. Let $G_{I,J}$ be the submatrix of $G$ restricted to rows indiced by $I$ and columns indiced by $J$. Then
\begin{equation}\label{eqn:norm_decomposition}
\|AGB\|_p \leq \sum_{i,j} \norm{A_{I_i}G_{I_i J_j}B_{J_j}}_p \leq \sum_{i,j} \frac{1}{2^{i+j}}\norm{G_{I_i J_j}}_p.
\end{equation}
Now, for each $i$ and $j$, it holds with probability $\geq 1-\exp(-cK^2\max\{s_i,t_j\})\geq 1-\exp(-cK^2)$ that
\[
\frac{1}{2^{i+j}} \|G_{I_i J_j}\|_{p} \simeq \frac{1}{2^{i+j}}\cdot K(\sqrt{s_i}+\sqrt{t_j})\min\{s_i,t_j\}^{\frac1p}
\]
We claim that summands on the rightmost side of \eqref{eqn:norm_decomposition} with $\max\{i,j\}\geq 3\log n$ are negligible. Indeed, taking $K = \Theta(\sqrt{\log n})$, then
\[
\sum_{\max\{i,j\}\geq 3\log n}\!\!\!\frac{1}{2^{i+j}} \|G_{I_i J_j}\|_{p} \leq \frac{1}{n^3} \cdot \Theta(\sqrt{\log n})\!\!\!\!\!\!\sum_{\max\{i,j\}\geq 3\log n} \!\!\!\!\!\!(s_i + t_j) \leq \frac{1}{n^2}\cdot \Theta(\sqrt{\log n}) \cdot 2n = o(1)
\]
with failure probability $\leq n^2\exp(-cK^2) = O(1)$. Note that when $i = 0$ and $j = 0$, the corresponding summand is $\gtrsim K = \Theta(\sqrt{\log n})$, hence the summands with $\max\{i,j\} \geq 3\log n$ is indeed negligible.

The claim result follows immediately, 
where we need to take a union bound over all $i,j\leq 3\log n$, so we need
\[
(3\log n + 1)^2\exp(-cK^2) = O(1),
\]
which holds when $K = \Theta(\log \log n)$.
\end{proof}

\begin{lemma}\label{lem:AGHB} Let $A$ and $B$ be deterministic $n\times N$ matrices and $G, H$ be $N\times r$ Gaussian random matrix of i.i.d.\ $N(0,1)$ entries. Suppose that $n\leq cr$ for some absolute constant $c \in (0,1)$. It holds with probability $1-O(1)$ that
\[
\|AGH^TB^T\|_p \gtrsim \sqrt{r} \|A\|_{op} \|B\|_{op} E_p(A, B),
\]
where $E_p(A,B)$ is as defined in \eqref{eqn:E_p}.
\end{lemma}
\begin{proof}
As in the proof of Lemma~\ref{lem:AXB}, we assume that $\|A\|_{op}=\|B\|_{op}=1$ and define $I_i$, $J_j$, $s_i$, $t_j$, $G_{I,J}$ in the same manner. Similarly to before, it holds with probability $1-O(1)$ that
\[ 
\sum_{\max\{i,j\}\geq 3\log n} \norm{A_{I_i}G_{I_i J_j}B_{J_j}}_p  = o(1).
\]
Now we choose the $(i,j)$ block with biggest Schatten-$p$ norm among $i,j\leq 3\log n$, that is, we choose $i$ and $j$ such that
\[
\norm{A_{I_i}G_{I_iJ_j}B_{J_j}}_p = \max_{1\leq i',j'\leq 3\log n} \norm{A_{I_{i'}}G_{I_{i'} J_{j'}}B_{J_{j'}}}_p.
\]
Then
\[
\norm{G_{I_i J_j}}_p\geq \min\{s_i, t_j\}^{\frac 1p} \left((\sqrt r-C\sqrt{s_i})(\sqrt r-C\sqrt{t_j})\right)^{\frac 1p} \gtrsim \min\{s_i, t_j\}^{\frac 1p} r
\]
with probability $\geq 1-\exp(-cr)$, and thus
\begin{align*}
\norm{A_{I_i}G_{I_i}H_{J_j}B_{J_j}}_p &\geq \frac{1}{2^{i+1}}\cdot\frac{1}{2^{j+1}}\cdot \norm{G_{I_i J_j}}_p\\
&\gtrsim \min\{s_i, t_j\}^{\frac1p} \frac{1}{2^{i}}\cdot\frac{1}{2^{j}} r\\
&\gtrsim \min\{s_i, t_j\}^{\frac1p} \max\{\sqrt{s_i},\sqrt{t_j}\}\frac{1}{2^{i}}\cdot\frac{1}{2^{j}} \sqrt{r}.
\end{align*}
The claimed lower bound follows immediately, noting that the sum over $\max\{i,j\}\geq 3\log n$ is negligible compared with the term corresponding to $i=j=0$.
\end{proof}

\begin{proof}[Proof of Theorem~\ref{thm:two_sided_p<2}']
Without loss of generality, we can assume that the maximum column norm of $R$ and that of $S$ are the same; otherwise we can rescale $R$ and $S$.

Let $r=n/(\rho^2D)$ and $t_0 = \theta r$ for some $\rho = \Theta(\log^{3p/(2-p)}t)$ and $\theta\in (0,1)$ to be determined. We shall show that if $t\leq t_0$, it will not happen that $G$ satisfies \eqref{eqn:two_sided_distortion'} for all $A\in T$.

Let $\mathcal{D}$ be the distribution of Gaussian random matrices of dimension $n\times r$ with i.i.d.\ entries $N(0,1)$ and let $G,H\sim \mathcal{D}$ be independent. It follows from Lemma~\ref{lem:AGHB} that with probability $\geq 1-O(1)$,
\begin{equation}\label{eqn:complicated_1}
\|\Sigma_R G H^T \Sigma_S^T\|_q\gtrsim \sqrt{r} E_q(R,S).
\end{equation}
On the other hand, it follows from \eqref{eqn:two_sided_distortion} that with probability $\geq 1-\exp(-c_1n)$,
\begin{equation}\label{eqn:complicated_2}
\|\Sigma_R G H^T \Sigma_S^T\|_q \leq D^{\frac12-\frac1p}\|G H^T\|_p \lesssim D^{\frac12-\frac1p} n r^{\frac 1p}.
\end{equation}
Now, let $\mathcal{F}$ be the distribution of $n\times n$ Gaussian matrix of i.i.d.\ entries $N(0,1)$ and let $F$ be drawn from $\mathcal{F}$. Then $\norm{RFS}_q$ is identically distributed as $\Sigma_R F' \Sigma_S$, where $F'$ is a random $t\times t$ Gaussian matrix of i.i.d.\ entries $N(0,1)$. It follows from Lemma~\ref{lem:AXB} that with probability $\geq 1-O(1)$,
\begin{equation}\label{eqn:complicated_3}
\|\Sigma_R F' \Sigma_S^T\|_q\lesssim (\log^{\frac 52}t)(\log \log t)E_q(R,S) \leq (\log^3 t)E_q(R,S)
\end{equation}
On the other hand, it follows from \eqref{eqn:two_sided_distortion} that with probability $\geq 1-\exp(-c_2n)$,
\begin{equation}\label{eqn:complicated_4}
\|R F S^T\|_q\geq \|F\|_p \gtrsim n^{1/p}\sqrt{n}.
\end{equation}

Define events $\p_1(G,H,R,S)$ and $\p_2(F,R,S)$ to be \eqref{eqn:complicated_1} and \eqref{eqn:complicated_3} respectively. Further define
\begin{align*}
\p_3(G,H):\ & \norm{GH^T}_p \lesssim nr^{1/p},\\
\p_4(F):\ & \norm{F}_p \lesssim n^{1/p}\sqrt{n}.
\end{align*}
Both $\p_3(G,H)$ and $\p_4(F)$ hold with probability $\geq 1-\exp(-c_3n)$ when $G,H\sim \mathcal{D}$ and $F\sim\mathcal{F}$.

Now, for $2m$ samples $G_1,\dots,G_m,H_1,\dots,H_m$ independently drawn from $\mathcal{D}$, and $m$ samples $F_1,\dots,F_m$ independently drawn from $\mathcal{F}$, it holds for any fixed $S$ and $T$ that
\begin{equation}\label{eqn:two_sided_pre-net}
\Pr_{G_i,H_i,F_i}\left\{\exists i\text{ s.t. }
\p_1(G_i,H_i,R,S)\text{ and }\p_2(F_i,R,S)\text{ and }\p_3(G_i,H_i)\text{ and }\p_4(F_i)\text{ all hold}
\right\} \geq 1-e^{-c_4m}.
\end{equation}

Since $1\leq \|Re_ie_j^TS^T\|_q = \|R_i\|_2\|S_j\|_2 \leq D$, we can restrict the matrix $R$ and $S$ to matrices with column norm in $[1,\sqrt{D}]$. Thus we can find a net $\mathcal{M}\subset \bigcup_{t=1}^{t_0} \R^{t\times n}$ of size $\exp(O(t_0n\ln(Dn/\eta))$ such that for any $M$ with column norms in $[1,\sqrt{D}]$, we can find $M'\in \mathcal{G}$ such that $\|M-M'\|_{op}\leq \eta$.

Now it follows from \eqref{eqn:two_sided_pre-net} that
\begin{multline*}
\Pr_{G_i,H_i,F_i}\left\{\forall R,S\in\mathcal{M}, \exists i,\ 
\p_1(G_i,H_i,R,S)\text{ and }\p_2(F_i,R,S)\text{ and }\p_3(G_i,H_i)\text{ and }\p_4(F_i)\text{ all hold}
\right\}\\
 \geq 1-\exp\left(O\left(t_0n\ln\frac{Dn}{\eta}\right)\right)\exp\left(-c_4m\right) > 0,
\end{multline*}
if we choose $m = \Theta(n\ln(Dn))$. Fix $\{G_i,H_i,F_i\}_i$ such that for each pair $R',S'\in\mathcal{M}$ there exists $i$ such that $\p_1(G_i,H_i,R',S')$ and $\p_2(F_i,R',S')$ and $\p_3(G_i,H_i)$ and $\p_4(F_i)$ all hold.

Take $T=\{I_n\}\cup\{e_ie_j^T\}\cup\{G_iH_i^T\}_i\cup\{F_i\}_i$. We know that if $(R,S)$ satisfies \eqref{eqn:two_sided_distortion} for all $A\in T$, then there exists $R'$ and $S'$ such that $\|R'-R\|_{op}\leq \eta$ and $\|S'-S\|_{op}\leq \eta$, and there exists $1\leq i\leq m$ such that $\p_1(G_i,H_i,R',S')$ and $\p_2(F_i,R',S')$ and $\p_3(G_i,H_i)$ and $\p_4(F_i)$ all hold. One can then show that \eqref{eqn:complicated_1}, \eqref{eqn:complicated_2}, \eqref{eqn:complicated_3}, \eqref{eqn:complicated_4} all hold with slightly larger or smaller constants for $R$ and $S$. It follows that
\[
\frac{n^{\frac1p}\sqrt{n}}{\log^3 t}\lesssim D^{\frac1p-\frac12}{\sqrt r}nr^{\frac 1p},
\]
or,
\[
\frac{1}{\log^3 t}\lesssim \left(\frac{rD}{n}\right)^{\frac 1p-\frac12} = \frac{1}{\rho^{\frac2p-1}},
\]
which contradicts our choice of $\rho$ (the hidden constant in $\lesssim$ above depends only on $D_0$, $\theta$ and $\eta$, and then we can choose the hidden constant in the $\Theta$-notation for $\rho$).
\end{proof}

Using almost the exact argument with the identical set $T$ as in the proof of Theorem~\ref{thm:two_sided_p<2} we can prove a similar bound for $p>2$. We need a lower bound for Lemma~\ref{lem:AXB} and an upper bound for Lemma~\ref{lem:AGHB}, which are corollaries of non-commutative Khintchine inequality. 

\begin{theorem} \label{thm:two_sided_p>2}
Let $p>2$. There exist a set $T\subset \R^{n\times n}$ with $|T|=\poly(n)$ and an absolute constant $c\in (0,1)$ such that,  if it holds for some matrices $R, S\in \R^{t\times n}$ with $t\leq cn$ and for all $A\in T$ that
\begin{equation}
\norm{A}_p \leq \norm{RAS^T}_q \leq D_{p,q}\norm{A}_p \tag{\ref{eqn:two_sided_distortion}}
\end{equation}
it must hold that
\[
D_{p,q} \gtrsim (n/t)^{\frac12-\frac1p}.
\]
\end{theorem}

The proof of the theorem is omitted but we shall show the two auxiliary corollaries of non-commutative Khintchine inequality. 

\begin{corollary}\label{lem:AXB2}
Let $A$ and $B$ be deterministic $n\times n$ matrices and $G$ be a Gaussian random matrix of i.i.d.\ $N(0,1)$ entries. It holds with probability $1-O(1)$ that
\[
\|AGB\|_p\simeq \max\{\|A\|_p\|B\|_F, \|A\|_F\|B\|_p \},\quad p\geq 2.
\]
\end{corollary}
\begin{proof}
Note that the function $X\mapsto \|AGB\|_p$ is a Lipschitz function with Lipschitz constant $\|A\|_{op} \|B\|_{op}$ when $p\geq 2$. Using concentration inequality for Lipschitz function on Gaussian space, it suffices to show that
\[
\E_G \|AGB\|_p\simeq \max\{\|A\|_p\|B\|_F, \|A\|_F\|B\|_p \}.
\]
By rotational invariance, assume that $A$ and $B$ are diagonal matrices, whose diagonal entries are $(a_1,\dots,a_n)$ and $(b_1,\dots,b_n)$, respectively.

Write \[
AGB = \sum_{i,j} g_{ij} a_ib_j e_ie_j^T =: \sum_{i,j} g_{ij} C_{i,j}.
\]
To apply noncommutative Khintchine inequality, we shall calculate 
\[
\norm{\left(\sum_{i,j} C_{i,j}C_{i,j}^T\right)^{\frac12}}_p\qquad\text{and}\qquad\norm{\left(\sum_{i,j} C_{i,j}^TC_{i,j}\right)^{\frac12}}_p.
\]
Note that
\begin{align*}
\sum_{i,j} C_{i,j}C_{i,j}^T = \sum_{i,j} a_i^2 b_j^2 e_i e_i^T = \|B\|_F^2 A^2\\
\sum_{i,j} C_{i,j}^T C_{i,j} = \sum_{i,j} a_i^2 b_j^2 e_j e_j^T = \|A\|_F^2 B^2,
\end{align*}
and thus
\[
\norm{\left(\sum_{i,j} C_{i,j}C_{i,j}^T\right)^{\frac12}}_p = \|B\|_F \|A\|_p\qquad\text{ and }\qquad \norm{\left(\sum_{i,j} C_{i,j}^T C_{i,j}\right)^{\frac12}}_p = \|A\|_F \|B\|_p.
\]
The result follows immediately from noncommutative Khintchine inequality.
\end{proof}

\begin{corollary}\label{lem:AGHB2} Let $A$ and $B$ be deterministic $n\times N$ matrices and $G, H$ be $N\times r$ Gaussian random matrix of i.i.d.\ $N(0,1)$ entries. Suppose that $n\leq cr$ for some absolute constant $c \in (0,1)$. It holds with probability $1-\exp(-r)$ that
\[
\|AGH^TB^T\|_p \simeq \sqrt{r}\max\{\|A\|_p\|B\|_F, \|A\|_F\|B\|_p\}.
\]
\end{corollary}
\begin{proof}
Note that the function $X\mapsto \|AGH^TB^T\|_p$ is a Lipschitz function with Lipschitz constant $\frac{1}{2}\|A\|_{op} \|B\|_{op}$ when $p\geq 2$. Using the concentration inequality for Lipschitz function on Gaussian space, it suffices to show that
\[
\E_{G,H} \|AGH^TB^T\|_p \simeq \sqrt{r}\max\{\|A\|_p\|B\|_F, \|A\|_F\|B\|_p\}.
\]
It follows from the proof of the previous lemma that
\[
\E_G \|AGH^TB^T\|_p \simeq \sqrt{r}\max\{\|A\|_p\|HB\|_F, \|A\|_F\|HB\|_p\}.
\]
The result is immediate, noting that $\E \|HB\|_p \simeq \sqrt{r}\|B\|_p$ because $\E s_{\max}(H), \E s_{\min}(H)\simeq \sqrt{r}$.
\end{proof}

%% file: upperbound.tex
\section{Upper bounds}\label{sec:ub}
We show specific designs of $G$ that achieve the distortion in $\tilde{D}_{p,q}$ advertised in the introduction, up to logarithmic factors. Specifically, we show that 
\begin{enumerate}[label=(\alph*),itemsep=0ex]
\item we can design $G$ with $r=1$ which attains the distortion $\tilde{D}_{p,q}$ for $1\leq p\leq 2\leq q$ and $1\leq q\leq 2\leq p$;
\item for some $\theta\in (0,1)$ and $t\leq \theta n$, we can design $G$ with $r = \Theta(t\log(n/t))$ rows which attains the distortion $\tilde{D}_{p,q}$, in all other cases of $p,q$.
\end{enumerate}

\subsection{Cases other than $1\leq p\leq 2\leq q$ and $1\leq q\leq 2\leq p$}
Let $G\in \R^{r\times n}$ ($r \geq C t$) be a random matrix and $c, c', \eta > 0$ be absolute constants which satisfy the following properties:
\begin{enumerate}[label=(\alph*),itemsep=0ex]
\item (subspace embedding) For a fixed $t$-dimensional subspace $X\subseteq \R^n$ it holds with probability $\geq 1-\exp(-c't)$ that
\[
(1-\eta)\norm{x}_2 \leq \norm{Gx}_2\leq (1+\eta)\norm{x}_2,\quad \forall x\in X.
\]
\item For a fixed $A\in \R^{n\times n}$ it holds with probability $\geq 1-\exp(-c'r)$ that
\[
\norm{GA}_{op} \leq c\left(\norm{A}_{op} + \frac{1}{\sqrt{r}}\norm{A}_{F}\right)
\]
\item For a fixed $A\in \R^{n\times n}$ it holds with probability $\geq 1-\exp(-c'r)$ that
\[
(1-\eta)\norm{A}_F\leq \norm{GA}_F \leq (1+\eta)\norm{A}_F.
\]
\end{enumerate}

Consider the singular value decomposition $A = U\Sigma V^T$, where $U$ and $V$ are orthogonal matrices, $\Sigma = \diag\{\sigma_1,\dots,\sigma_n\}$ with $\sigma_1\geq \sigma_2\geq \cdots$. For an index set $I\subseteq [n]$, define $A_I = U\Sigma_I V^T$, where $\Sigma_I$ is $\Sigma$ restricted to the diagonal elements with indices inside $I$ (the diagonal entries with indices outside $I$ are replaced with $0$).

\begin{theorem}\label{thm:1<=p<2}
Let $p,q\geq 1$. There exist constants $\theta=\theta(p,q) < 1$ small enough and $C=C(p,q)$ large enough such that for $t\leq \theta n$ and matrix $G$ satisfying the aforementioned properties, it holds for any (fixed) $A\in \R^{n\times n}$ with probabilty $1-\exp(-c''t)$ that
\begin{gather*} 
\frac{t^{\frac1q-\frac12}}{n^{\frac1p-\frac12}\log\frac{n}{t}}\norm{A}_p \lesssim \norm{GA}_q\lesssim \|A\|_p,\quad p \leq q < 2\\
\min\left\{\frac{1}{\log\frac{n}{t}},\frac{t^{\frac1q-\frac12}}{n^{\frac1p-\frac12}}\right\}\norm{A}_p \lesssim \norm{GA}_q\lesssim t^{\frac1q-\frac1p}\|A\|_p,\quad q \leq p < 2\\
\frac{1}{t^{\frac1p-\frac1q}\log\frac{n}{t}}\norm{A}_p \lesssim \norm{GA}_q\lesssim \max\left\{\frac{n^{\frac 12-\frac 1p}}{t^{\frac12-\frac1q}},1\right\}\|A\|_p,\quad q \geq p \geq 2\\
\frac{1}{\log\frac{n}{t}}\norm{A}_p \lesssim \norm{GA}_q\lesssim \frac{n^{\frac 12-\frac 1p}}{t^{\frac12-\frac1q}}\|A\|_p,\quad p \geq q \geq 2
\end{gather*}
\end{theorem}
Note that for $t = \Omega(\log n)$ and $r = Ct$ for some large $C$, a Gaussian random matrix of i.i.d.\ entries $N(0,1/r)$, or a randomized Hadamard Transform matrix of $t \polylog(t)$ rows, satisfies the conditions on $G$~\cite{CNW16}. We thus have an immediate corollary of Theorem~\ref{thm:1<=p<2} as follows.
\begin{corollary}\label{cor:1<=p<2}
Suppose that $1\leq p,q$ and $c\log n \leq t\leq \theta n$ for some absolute constants $\theta\in(0,1)$ and $c\geq 1$. There exists (random) $G\in \R^{r\times m}$ with $r\gtrsim t$ such that with probability $\geq 1-\exp(-c''t)$,
\begin{gather*}
\norm{A}_p \leq \norm{GA}_q \lesssim \frac{n^{\frac1p-\frac12}}{t^{\frac1q-\frac12}}\left(\log\frac{n}{t}\right)\norm{A}_p,\quad 1\leq p< q< 2;\\
\norm{A}_p \leq \norm{GA}_q \lesssim \max\left\{\left(\frac{n}{t}\right)^{\frac1p-\frac12},t^{\frac1q-\frac1p}\right\}\left(\log\frac{n}{t}\right)\|A\|_p,\quad 1\leq q\leq p\leq 2;\\
\norm{A}_p \leq \norm{GA}_q \lesssim \max\left\{\left(\frac{n}{t}\right)^{\frac12-\frac1p},t^{\frac1p-\frac1q}\right\}\left(\log\frac{n}{t}\right)\|A\|_p,\quad q\geq p\geq 2;\\
\norm{A}_p \lesssim \norm{GA}_q\lesssim  \frac{n^{\frac 12-\frac 1p}}{t^{\frac12-\frac1q}}\left(\log\frac{n}{t}\right)\|A\|_p,\quad p \geq q \geq 2.
\end{gather*}
In particular when $p=q$,
\begin{gather*}
\norm{A}_p \leq \norm{GA}_p \lesssim \left(\frac{n}{t}\right)^{\frac1p-\frac12}\left(\log\frac{n}{t}\right),\quad p \leq 2\\
\norm{A}_p \leq \norm{GA}_p \lesssim \left(\frac{n}{t}\right)^{\frac12-\frac1p}\left(\log\frac{n}{t}\right),\quad p \geq 2
\end{gather*}
\end{corollary}

To prove Theorem~\ref{thm:1<=p<2}, we need a few auxillary lemmata.

\begin{lemma}\label{lem:conditions_ub}
Let $\theta$, $t$, $C$ and $G$ be as defined in Theorem~\ref{thm:1<=p<2} and $b=\Theta(\log(n/t))$. At least one of the following conditions will hold:
\begin{equation}\label{eqn:large_ky_fan}
\sum_{i=1}^{b t} \sigma_i^p \geq \frac{1}{2}\sum_{i=1}^n \sigma_i^p.
\end{equation}
and
\begin{equation}\label{eqn:stop}
\sigma_s^2 \leq \frac{2}{t} \sum_{i=s}^n \sigma_i^2 \quad\text{for some}\quad s\leq b t.
\end{equation}
\end{lemma}

To prove the preceding lemma we need a further auxiliary lemma. Consider the first $b$ blocks of singular values of $A$ each of size $t$, that is, $I_1=\{\sigma_1,\dots,\sigma_{t}\}$, $\dots$, $I_b= \{\sigma_{(b-1)t+1},\dots,\sigma_{bt}\}$.
\begin{lemma}\label{lem:condition_aux}
If \eqref{eqn:stop} does not hold for any $s\leq  b t$, it must hold for all $2\leq j\leq b$ that $\sigma_{j t}\leq \frac{1}{2}\sigma_{(j-1) t}$.
\end{lemma}
\begin{proof}
If this is not true for some $j$ then
\[
\sum_{i=(j-1) t+1}^ {j t} \sigma_i^2 \geq  t\sigma_{j t}^2 > \frac{ t}{2}\sigma_{(j-1) t}^2,
\]
which contradicts \eqref{eqn:stop} with $s=(j-1)t \leq b t$.
\end{proof}

Now we prove Lemma~\ref{lem:conditions_ub}.
\begin{proof}[Proof of Lemma~\ref{lem:conditions_ub}]
Suppose that \eqref{eqn:large_ky_fan} does not hold and we need to show that \eqref{eqn:stop} holds for some $s\leq  b t$. Otherwise, it follows from the Lemma~\ref{lem:condition_aux} that
\[
\sigma_{b t+1} \leq \frac{\sigma_1}{2^b} \leq \left(\frac{t}{n}\right)^2\sigma_1
\]
and thus
\begin{equation}\label{eqn:ky_fan_tail_ub}
\sum_{i=b t+1}^n \sigma_i^p < n\sigma_{b t+1}^p \leq \frac{t^{2p}}{n^{2p-1}}\sigma_1^p\leq t\theta^{2p-1}\sigma_1^p,
\end{equation}
On the other hand,
\begin{equation}\label{eqn:ky_fan_lb}
\sum_{i=1}^{b t} \sigma_i^p \geq t\sigma_1^p\left(\frac{1}{2} + \frac{1}{4} + \cdots + \frac{1}{2^b}\right) = \left(1-\frac{1}{2^b}\right)t\sigma_1^p = (1-\theta^2)t\sigma_1^p.
\end{equation}
Using the assumption on $\theta$, we see that the rightmost side of \eqref{eqn:ky_fan_lb} is bigger than the rightmost side of \eqref{eqn:ky_fan_tail_ub}, which contradicts the assumption that \eqref{eqn:large_ky_fan} does not hold.\end{proof}

\begin{lemma}\label{lem:key_alg_lb}
Let $p,q\geq 1$, and $t$, $b$ and $G$ be defined as in Lemma~\ref{lem:conditions_ub}. Suppose that $s$ satisfies \eqref{eqn:stop} and let $J=\{s,s+1,\dots,n\}$. Then
\[
\norm{GA_J}_q \gtrsim 
\begin{cases}
\frac{t^{\frac1q-\frac12}}{n^{\frac1p-\frac12}}\norm{A_J}_p, & p\leq 2\\
\frac{1}{t^{\frac1p-\frac1q}}\norm{A_J}_p, & p > 2.
\end{cases}
\]
\end{lemma}
\begin{proof}
Combining Property (b) of $G$ with \eqref{eqn:stop} yields that
\[
\|GA_J\|_{op} \leq \frac{c}{\sqrt t}\left(\sqrt{2}+\sqrt{\frac{1}{C}}\right)\|A_J\|_F =: \frac{K}{\sqrt{t}}\|A_J\|_F
\]
On the other hand,
\[
\|GA_J\|_F \geq \frac{1}{2}\|A_J\|_F.
\]
This implies that at least $\alpha r$ singular values of $GA_J$ are at least $\frac{\gamma}{\sqrt{t}}\|A_J\|_F$, provided that
\[
C\left((1-\alpha)\gamma^2 + \alpha K^2\right) < \frac{1}{4},
\]
which is satisfied if we choose $\gamma = \Theta(1/\sqrt{C})$, $\alpha = \Theta(1/K^{2/q})$.

Now, when $p\leq 2$,
\begin{align*}
\norm{GA_J}_q \geq (\alpha r)^{\frac{1}{q}} \frac{\gamma}{\sqrt{t}}\norm{A_J}_F \geq (\alpha C)^{\frac{1}{q}}\gamma \cdot \frac{t^{\frac1q-\frac12}}{n^{\frac1p-\frac12}}\norm{A_J}_p.
\end{align*}
When $p > 2$, we have 
\[
\norm{A_J}_p \leq \norm{A_J}_{op}^{1-\frac2p} \norm{A_J}_F^{\frac 2p} \leq \left(\frac{2}{ t}\right)^{\frac12-\frac1p}\|A_J\|_F,
\]
and thus
\[
\norm{GA_J}_q \geq (\alpha r)^{\frac{1}{q}} \frac{\gamma}{\sqrt{t}}\norm{A_J}_F \geq (\alpha C)^{\frac{1}{q}}\frac{\gamma}{t^{\frac12-\frac1q}}\cdot \frac{1}{\left(\frac{2}{ t}\right)^{\frac12-\frac1p}}\norm{A_J}_p\gtrsim \frac{1}{t^{\frac1p-\frac1q}}\norm{A_J}_p.\qedhere
\]
\end{proof}

\begin{lemma}\label{lem:key_alg_ub}
Let $p,q\geq 1$, and $t$, $b$ and $G$ be defined as in Lemma~\ref{lem:conditions_ub}. Suppose that $s$ satisfies \eqref{eqn:stop} and let $J=\{s,s+1,\dots,n\}$. Then
\[
\norm{GA_J}_q \lesssim 
\begin{cases}
\frac{1}{t^{\frac1p-\frac1q}}\norm{A_J}_p, & p, q \leq 2;\\
\frac{n^{\frac12-\frac1p}}{t^{\frac12-\frac1q}}\norm{A_J}_p, & p,q\geq 2.
\end{cases}
\]
\end{lemma}
\begin{proof}
When $p \leq 2$, we have that
\[
\norm{A_J}_F^2 \leq \norm{A_J}_p^p \norm{A_J}_{op}^{2-p}.
\]
Using \eqref{eqn:stop}, we obtain that
\[
\norm{A_J}_p \geq \frac{\norm{A_J}_F^{2/p}}{\norm{A_J}_{op}^{2/p-1}} \geq \left(\frac{ t}{2}\right)^{\frac1p-\frac12}\norm{A_J}_F.
\]
When $q \leq 2$, it follows from Property (c) of $G$ that
\[
\|GA_J\|_q \leq r^{\frac{1}{q}-\frac{1}{2}}\|GA_J\|_F \leq (1+\eta)r^{\frac{1}{q}-\frac{1}{2}}\|A_J\|_F.
\]
Thus when $p, q\leq 2$,
\[
\|GA_J\|_q \leq (1+\eta) r^{\frac1q-\frac12} \left(\frac{2}{ t}\right)^{\frac 1p-\frac12} \norm{A_J}_p
= (1+\eta) (Cb q t)^{\frac1q-\frac12} \left(\frac{2}{ t}\right)^{\frac 1p-\frac12} \norm{A_J}_p
\lesssim \frac{1}{t^{\frac 1p-\frac1q}}\norm{A_J}_p.
\]
When $p,q >2$,
\begin{multline*}
\norm{GA_J}_q 
\leq \norm{GA_J}_{op}^{1-\frac2q} \norm{GA_J}_F^{\frac 2q}
\leq \left(\frac{K}{\sqrt t}\norm{A}_F\right)^{1-\frac2q}\left((1+\eta)\norm{A}_F\right)^{\frac 2q}\\
= \left((1+\eta)\right)^{\frac 2q} \frac{1}{t^{\frac12-\frac1q}}\norm{A}_F
\lesssim \frac{n^{\frac 12-\frac 1p}}{t^{\frac12-\frac1q}}\norm{A}_p.
\qedhere
\end{multline*}
\end{proof}

Now we are ready to show Theorem~\ref{thm:1<=p<2}.
\begin{proof}[Proof of Theorem~\ref{thm:1<=p<2}]
It follows from the subspace embedding property of $G$ to show that
\[
(1-\eta)\norm{A_{I_i}}_q\leq \norm{GA_{I_i}}_q\leq (1+\eta)\norm{A_{I_i}}_q,\quad 1\leq i\leq b
\]
and thus
\begin{align*}
\frac{1-\eta}{t^{\frac1p-\frac1q}}\norm{A_{I_i}}_p\leq \norm{GA_{I_i}}_q\leq (1+\eta)\norm{A_{I_i}}_p, \qquad p \leq q;\\
(1-\eta)\norm{A_{I_i}}_p\leq \norm{GA_{I_i}}_q\leq (1+\eta)t^{\frac1q-\frac1p}\norm{A_{I_i}}_p, \qquad p \geq q.
\end{align*}
When \eqref{eqn:large_ky_fan} holds, there exists $i^\ast$ ($1\leq i^\ast\leq b$) such that
\[
\norm{A_{I_{i^\ast}}}_p\geq \frac{1}{2^{\frac1p}b} \norm{A}_p
\]
and thus
\begin{gather*}
\frac{1}{bt^{\frac1p-\frac1q}}\norm{A}_p\lesssim \norm{GA_{I_{i^\ast}}}_q\lesssim \norm{A}_p, \quad p\leq q\\
\frac{1}{b}\norm{A}_p\lesssim \norm{GA_{I_{i^\ast}}}_q\lesssim t^{\frac1q-\frac1p}\norm{A}_p, \quad p\geq q
\end{gather*}
When \eqref{eqn:large_ky_fan} does not hold, let $J$ be as defined in Lemma~\ref{lem:key_alg_lb} and
\[
\frac{1}{2^{\frac1p}}\norm{A}_p \leq \norm{A_J}_p\leq \norm{A}_p.
\]
The claimed upper and lower bounds follow from combining the bounds above, together with Lemma~\ref{lem:key_alg_lb}, Lemma~\ref{lem:key_alg_ub}, and 
\begin{gather*}
\norm{GA}_q \geq \max\left\{ \norm{GA_{I_1}}_q,\dots,\norm{GA_{I_b}}_q, \norm{GA_J}_q \right\}\\
\norm{GA}_q \leq \sum_{i=1}^b \norm{GA_{[I_i]}}_q + \norm{GA_J}_q.\qedhere
\end{gather*}
\end{proof}

\subsection{Case $1\leq p\leq 2\leq q$ and $1\leq q\leq 2\leq p$}
\begin{theorem}\label{thm:p>2_ub}
Let $1\leq p\leq 2\leq q$ or $1\leq q\leq 2\leq p$. Let $g\sim N(0,I_n)$. Then with arbitrarily large constant probability, 
\begin{gather*}
n^{\frac 12-\frac 1p}\norm{A}_p \lesssim \norm{g^TA}_q\lesssim \|A\|_p,\quad 1\leq p\leq 2\leq q;\\
\norm{A}_p \lesssim \norm{g^TA}_q\lesssim n^{\frac 1p-\frac 12}\|A\|_p,\quad 1\leq p\leq 2\leq q.
\end{gather*}
\end{theorem}
\begin{proof}
Note that $g^TA$ is a row vector, and it holds that $\norm{g^TA}_q = \norm{g^TA}_F$ for all $q$. 

Since $\norm{g^TA}_F = \Theta(\|A\|_F)$ with arbitrarily large constant probability, the theorem follows from the facts that
\begin{gather*}
n^{\frac12-\frac1p}\leq \|A\|_F\leq \|A\|_p,\quad p < 2\\
\|A\|_p\leq \|A\|_F\leq n^{\frac12-\frac1p}\|A\|_p,\quad p > 2.\qedhere
\end{gather*}
\end{proof}

%% file: streaming.tex
\section{Application to Streaming Algorithms}\label{sec:streaming}
Here we show that we can implement our embedding as a streaming algorithm.
The two things we need to show are that our sketching matrices can be maintained
with limited randomness, and our sketch can be maintained in small space. 

We first prove our claim that we can truncate sketching matrices to $\Theta(\log n)$ bits for each entry by verifying that properties (a), (b) and (c) of
Section \ref{sec:ub} will continue to hold after truncation. Let $G$ be a matrix satisfying properties (a), (b) and (c), and let $G'$ be of the same dimension as $G$ such that $(G-G')_{ij}\leq 1/\poly(n)$ for all $i, j$. We can choose the power of $n$ in $\poly(n)$ big enough such that $\norm{G-G'}_{op}\leq 0.001$. Then 
\[
\left|\|G'x\|_2 - \|Gx\|_2\right| \leq \|(G-G')x\|_2 \leq \|G-G'\|_{op}\|x\|_2 \leq 0.001\|x\|_2
\]
and
\[
\left|\norm{G'A}_{F} - \norm{GA}_{F}\right| \leq \norm{(G-G')A}_{F} \leq \norm{G-G'}_{op}\norm{A}_{F} \leq 0.001\norm{A}_{F},
\]
which shows that properties (a) and (c) hold for $G'$ with a slightly bigger $\eta$. Lastly,
\[
\left|\norm{G'A}_{op} - \norm{GA}_{op}\right| \leq \norm{(G-G')A}_{op} \leq \norm{G-G'}_{op}\norm{A}_{op} \leq 0.001\norm{A}_{op},
\]
which shows that property (b) holds for $G'$ with a slightly bigger constant $c$.

Next we show that the three properties (a), (b) and (c) also hold for Gaussian random matrices with reduced randomness. Let $G$ be a random matrix with $\Theta(r)$-wise independent entries each drawn from an $N(0,1/r)$ distribution, but truncated to additive $1/\poly(n)$ for a suficiently large $\poly(n)$ (recall that in our streaming application, $r = n/D^2$ -- see Section \ref{sec:intro} for discussion). It is known that without the truncation, $G$ with $\Theta(r)$-wise independent entries provides a subspace embedding for $t$-dimensional spaces (see, e.g., the second part of the proof of Theorem 8 of \cite{kvw14}, which is stated for sign matrices but the same argument holds for Gaussians. For the latter, one can replace Theorem 2.2 of \cite{cw09} with the more general
Theorem 6 and Remark 1 of \cite{kmn11}), which is property (a). It is also known that $G$ is a Johnson-Lindenstrauss transform~\cite[Remark 7]{KN10}, that is,
\[
\Pr\left\{ (1-\eta)\|x\|_2 \leq \|Gx\|_2 \leq (1+\eta)\|x\|_2\right\} \geq 1-\exp(-c'r),
\]
whence Property (c) follows immediately by taking a union bound over the columns $x$ of $A$ (recall that $r\geq t\geq \log n$). Since $G$ provides a subspace embedding, property (b) follows from Theorem 1 in \cite{CNW16}. Finally, note
that as argued above, given that properties (a), (b), and (c) hold for $G$
before truncation, they also hold after truncation.

It follows from the discussion above that we can store an $O((n/D^2) \log n)$
bit
seed to succinctly describe and generate matrices $R$ and $S$, and for
matrices $A$ specified with $O(\log n)$ bits, we can store our sketch
$RAS$ in a stream using $(n^2/D^4) \polylog(n)$ bits of memory. Note that
the space needed to store the random seed to generate $R$ and $S$ is
negligible compared to the space to store the sketch $RAS$.

%% file: main.bbl
\begin{thebibliography}{10}

\bibitem{a10}
Alexandr Andoni.
\newblock Nearest neighbor search in high-dimensional spaces.
\newblock In {\em the workshop: Barriers in Computational Complexity II}, 2010.
\newblock \url{http://www.mit.edu/~andoni/nns-barriers.pdf}.

\bibitem{akr15}
Alexandr Andoni, Robert Krauthgamer, and Ilya Razenshteyn.
\newblock Sketching and embedding are equivalent for norms.
\newblock In {\em Proceedings of the Forty-Seventh Annual ACM on Symposium on
  Theory of Computing}, pages 479--488. ACM, 2015.

\bibitem{BBCKY16}
Jaroslaw Blasiok, Vladimir Braverman, Stephen~R. Chestnut, Robert Krauthgamer,
  and Lin~F. Yang.
\newblock Streaming symmetric norms via measure concentration.
\newblock arXiv:1511.01111, 2016.

\bibitem{bcky16}
Vladimir Braverman, Stephen~R. Chestnut, Robert Krauthgamer, and Lin~F. Yang.
\newblock Sketches for matrix norms: Faster, smaller and more general.
\newblock {\em CoRR}, abs/1609.05885, 2016.

\bibitem{cr12}
Emmanuel Candes and Benjamin Recht.
\newblock Exact matrix completion via convex optimization.
\newblock {\em Communications of the ACM}, 55(6):111--119, 2012.

\bibitem{cw09}
Kenneth~L. Clarkson and David~P. Woodruff.
\newblock Numerical linear algebra in the streaming model.
\newblock In {\em Proceedings of the 41st Annual {ACM} Symposium on Theory of
  Computing, {STOC} 2009, Bethesda, MD, USA, May 31 - June 2, 2009}, pages
  205--214, 2009.

\bibitem{clarkson2013low}
Kenneth~L Clarkson and David~P Woodruff.
\newblock Low rank approximation and regression in input sparsity time.
\newblock In {\em Proceedings of the forty-fifth annual ACM symposium on Theory
  of computing}, pages 81--90. ACM, 2013.

\bibitem{CNW16}
Michael~B. Cohen, Jelani Nelson, and David~P. Woodruff.
\newblock Optimal approximate matrix product in terms of stable rank.
\newblock In {\em 43rd International Colloquium on Automata, Languages, and
  Programming, {ICALP} 2016, July 11-15, 2016, Rome, Italy}, pages 11:1--11:14,
  2016.

\bibitem{hms11}
Aram~W Harrow, Ashley Montanaro, and Anthony~J Short.
\newblock Limitations on quantum dimensionality reduction.
\newblock In {\em International Colloquium on Automata, Languages, and
  Programming}, pages 86--97. Springer, 2011.

\bibitem{kmn11}
Daniel~M. Kane, Raghu Meka, and Jelani Nelson.
\newblock Almost optimal explicit johnson-lindenstrauss families.
\newblock In {\em Approximation, Randomization, and Combinatorial Optimization.
  Algorithms and Techniques - 14th International Workshop, {APPROX} 2011, and
  15th International Workshop, {RANDOM} 2011, Princeton, NJ, USA, August 17-19,
  2011. Proceedings}, pages 628--639, 2011.

\bibitem{KN10}
Daniel~M Kane and Jelani Nelson.
\newblock A derandomized sparse {J}ohnson-{L}indenstrauss transform.
\newblock arXiv:1006.3585, 2010.

\bibitem{kvw14}
Ravi Kannan, Santosh Vempala, and David~P Woodruff.
\newblock Principal component analysis and higher correlations for distributed
  data.
\newblock In {\em COLT}, 2014.

\bibitem{kv16}
Weihao Kong and Gregory Valiant.
\newblock Spectrum estimation from samples.
\newblock {\em CoRR}, abs/1602.00061, 2016.

\bibitem{LN16}
Kasper~Green Larsen and Jelani Nelson.
\newblock {The Johnson-Lindenstrauss Lemma Is Optimal for Linear Dimensionality
  Reduction}.
\newblock In Ioannis Chatzigiannakis, Michael Mitzenmacher, Yuval Rabani, and
  Davide Sangiorgi, editors, {\em 43rd International Colloquium on Automata,
  Languages, and Programming (ICALP 2016)}, volume~55 of {\em Leibniz
  International Proceedings in Informatics (LIPIcs)}, pages 82:1--82:11, 2016.

\bibitem{LT91}
Michel Ledoux and Michel Talagrand.
\newblock {\em Probability in {B}anach spaces}.
\newblock Springer-Verlag, Berlin, 1991.

\bibitem{lnw14}
Yi~Li, Huy~L. Nguyen, and David~P. Woodruff.
\newblock On sketching matrix norms and the top singular vector.
\newblock In {\em Proceedings of the Twenty-Fifth Annual {ACM-SIAM} Symposium
  on Discrete Algorithms, {SODA} 2014, Portland, Oregon, USA, January 5-7,
  2014}, pages 1562--1581, 2014.

\bibitem{lw16a}
Yi~Li and David~P. Woodruff.
\newblock On approximating functions of the singular values in a stream.
\newblock In {\em Proceedings of the 48th Annual {ACM} {SIGACT} Symposium on
  Theory of Computing, {STOC} 2016, Cambridge, MA, USA, June 18-21, 2016},
  pages 726--739, 2016.

\bibitem{lw16b}
Yi~Li and David~P. Woodruff.
\newblock Tight bounds for sketching the operator norm, schatten norms, and
  subspace embeddings.
\newblock In {\em Approximation, Randomization, and Combinatorial Optimization.
  Algorithms and Techniques, {APPROX/RANDOM} 2016, September 7-9, 2016, Paris,
  France}, pages 39:1--39:11, 2016.

\bibitem{LP86}
Fran{\c{c}}oise Lust-Piquard.
\newblock In{\'e}galit{\'e}s de khintchine dans $c_p$ ($1<p<\infty$).
\newblock {\em Comptes Rendus de l'Académie des Sciences - Series I -
  Mathematics}, 303:289--292, 1986.

\bibitem{mm13}
Xiangrui Meng and Michael~W. Mahoney.
\newblock Low-distortion subspace embeddings in input-sparsity time and
  applications to robust linear regression.
\newblock In {\em Symposium on Theory of Computing Conference, STOC'13, Palo
  Alto, CA, USA, June 1-4, 2013}, pages 91--100, 2013.

\bibitem{MNSUW17}
Cameron Musco, Praneeth Netrapalli, Aaron Sidford, Shashanka Ubaru, and
  David~P. Woodruff.
\newblock Spectral sums beyond fast matrix multiplication: Algorithms and
  hardness.
\newblock manuscript, 2016.

\bibitem{nn13}
Jelani Nelson and Huy~L. Nguyen.
\newblock {OSNAP:} faster numerical linear algebra algorithms via sparser
  subspace embeddings.
\newblock In {\em 54th Annual {IEEE} Symposium on Foundations of Computer
  Science, {FOCS} 2013, 26-29 October, 2013, Berkeley, CA, {USA}}, pages
  117--126, 2013.

\bibitem{UCS16}
Shashanka Ubaru, Jie Chen, and Yousef Saad.
\newblock Fast estimation of $\tr({F}({A}))$ via stochastic lanczos quadrature.
\newblock 2016.
\newblock URL: \url{http://www-users.cs.umn.edu/~saad/PDF/ys-2016-04.pdf}.

\bibitem{V12}
Roman Vershynin.
\newblock Introduction to the non-asymptotic analysis of random matrices.
\newblock In Yonina~C. Eldar and Gitta Kutyniok, editors, {\em Compressed
  Sensing: Theory and Practice}, pages 210--268. Cambridge University Press,
  2012.

\bibitem{w05}
Andreas~J. Winter.
\newblock Quantum and classical message identification via quantum channels.
\newblock {\em Quantum Information {\&} Computation}, 5(7):605--606, 2005.

\bibitem{w14}
David~P. Woodruff.
\newblock Sketching as a tool for numerical linear algebra.
\newblock {\em Foundations and Trends in Theoretical Computer Science},
  10(1-2):1--157, 2014.

\end{thebibliography}
